 \documentclass[preprint,review,12pt]{elsarticle}

\usepackage[utf8x]{inputenc}
\usepackage[left=3cm, right=3cm, top=3cm, bottom=3cm]{geometry}
\usepackage{setspace}
\usepackage{comment}
\usepackage{wrapfig}
\usepackage{subfigure}
\usepackage[english]{babel}
\usepackage{csquotes}
\usepackage[ruled,vlined]{algorithm2e}
\usepackage{amssymb,amsmath,amsthm} % S’mbolos matem‡ticos
\usepackage{enumerate}
\usepackage{hyperref} % Hiperv’nculos en el ’ndice
\usepackage{graphicx}
\usepackage{amsfonts}
\usepackage{float}
\usepackage{fnpos} % Footnotes al final de p‡g.
\usepackage[usenames,dvipsnames]{xcolor} % Color
\usepackage{tabularx} % Big tables
\usepackage{adjustbox}
\usepackage{longtable}
\usepackage{lineno}
\usepackage{threeparttable} 
\usepackage{tikz}
\usetikzlibrary{trees}

\newtheorem{theorem}{Theorem}[section]
\newtheorem{lemma}[theorem]{Lemma}
\newtheorem{corollary}[theorem]{Corollary}

\tikzstyle{level 1}=[level distance=3.5cm, sibling distance=3.5cm]
\tikzstyle{level 2}=[level distance=3.5cm, sibling distance=2cm]

\tikzstyle{bag} = [text width=4em, text centered]
\tikzstyle{end} = [circle, minimum width=3pt,fill, inner sep=0pt]
\newcommand{\R}{\mathbb{R}}
\newcommand{\OO}{\mathcal{O}}
\newcommand{\SO}{\mathcal{O}^*}
\newcommand{\N}{\mathbb{N}}

\journal{Arxiv}

\begin{document}

\begin{frontmatter}

\title{Novel Matrix Hit and Run for Sampling Polytopes\\ and Its GPU Implementation}



%%%%%%%%%%%%%%%%%%%%%%%%%%%%%%%%%%%%%%%%%%%%%%%%%%%
%%%%%%%%%%%%%%%%%%%%%%%%%%%%%%%%%%%%%%%%%%%%%%%%%%%
%%%%%%%%%%%%%%%%%%%%%%%%%%%%%%%%%%%%%%%%%%%%%%%%%%%
%%%%%%%%%%%%%%%%%%%%%%%%%%%%%%%%%%%%%%%%%%%%%%%%%%%
%%%%%%%%%%%%%%%%%%%%%%%%%%%%%%%%%%%%%%%%%%%%%%%%%%%
\author[label1]{Mario Vazquez Corte}
\address[label1]{Department of Computer Science, Sonder.art, email \href{mailto:uumami@sonder.art}{uumami@sonder.art}}
\author[label2]{Luis V. Montiel}
\address[label2]{Department of Operations Research and Industrial Engineering.\\
Instituto Tecnol\'ogico Aut\'onomo de M\'exico - ITAM \\ R\'io Hondo 1, CDMX, M\'exico.}

%%%%%%%%%%%%%%%%%%%%%%%%%%%%%%%%%%%%%%%%%%%%%%%%%%%
%%%%%%%%%%%%%%%%%%%%%%%%%%%%%%%%%%%%%%%%%%%%%%%%%%%
%%%%%%%%%%%%%%%%%%%%%%%%%%%%%%%%%%%%%%%%%%%%%%%%%%%
%%%%%%%%%%%%%%%%%%%%%%%%%%%%%%%%%%%%%%%%%%%%%%%%%%%
%%%%%%%%%%%%%%%%%%%%%%%%%%%%%%%%%%%%%%%%%%%%%%%%%%%

\setstretch{1.1} 
\begin{abstract}

We propose and analyze a new Markov Chain Monte Carlo algorithm that generates a uniform sample over full and non-full dimensional polytopes. This algorithm, termed "Matrix Hit and Run" (MHAR), is a modification of the Hit and Run framework. For the regime $n^{1+\frac{1}{3}} \ll m$, MHAR has a lower asymptotic cost per sample in terms of soft-O notation ($\SO$) than do existing sampling algorithms after a \textit{warm start}. MHAR is designed to take advantage of matrix multiplication routines that require less computational and memory resources. Our tests show this implementation to be substantially faster than the \textit{hitandrun} R package, especially for higher dimensions. Finally, we provide a python library based on Pytorch and a Colab notebook with the implementation ready for deployment in architectures with GPU or just CPU. 
\end{abstract}


\end{frontmatter}

\setstretch{1.5} 


\section{Introduction} \label{S:1}

Random sampling of convex bodies is employed in disciplines such as operations research, statistics, probability, and physics. Among random-sampling approaches, Markov Chain Monte Carlo (MCMC) is the fastest, most accurate, and easiest to use \cite{walks}. MCMC is often implemented using polytope sampling algorithms, which are used in volume estimation \cite{polytope_volume} \cite{RHMC} \cite{practical_polytope}, convex optimization \cite{convex_program} \cite{optimization}, contingency tables \cite{tables}, mixed integer programming \cite{mix_int}, linear programming \cite{lin_pog}, hard-disk modeling \cite{hard_disk}, and decision analysis \cite{har_tommi} \cite{montiel_jd} \cite{montiel_b}.

Sampling methods start by defining a Markov chain whose stationary distribution converges to a desired target distribution. Then they draw a predetermined number of samples. These methods have two sources of computational complexity: \textit{mixing-time}, which is the number of samples needed to lose the ``dependency'' between each draw; and \textit{cost per iteration}, which is the number of operations required to obtain a single sample. Sampling algorithms aim for efficient mixing-times, so that they can produce independent samples without dropping (also called "burning") too many of them, and a low cost per iteration in order to draw samples fast \cite{mcmc_goals}.

%%%%%%%%%%%%%%%%%%%%%%%%%%%%%%%%%%%%%%%%%%%%%%%%%%%%%%%%%%%%%%%%%%%%
\subsection{History and relevance of MCMC}

The use of Monte Carlo methods has surged in the last 50 years, due to the availability of modern computers. However, there are records of experiments leading to a Monte Carlo simulation method as early as 1901 when Mario Lazzarini approximate $\pi$ by manually repeating Buffon's needle experiment 3,408 times. During the first half of the 20th century the use of Monte Carlo had a frequentist approach, since the Bayesian approach was viewed as unfavorable due to philosophical and computational considerations. With the advent of MCMC together with more powerful computers, Bayesian Monte Carlo methods saw an increase in use, having its first application published in 1993 with the ``the bootstrap filter'' \cite{first_bayesian}. 
 
Recently, numerous applications in operations research have used MCMC to complement diverse optimization models. For example, the characterization of a joint probability distribution under partial information is perhaps not unique \cite{montiel_b}. Hence, if we need the joint probabilities to value a real option \cite{montiel_a}, or to optimize the net gain of an oil field \cite{montiel_jd}, we have to understand the space to which the joint distribution belongs. Another example is the incomplete specification of a multi-attribute utility function in decision analysis. Here, the problem is to understand the range of preferences of the decision maker to provide recommendations \cite{har_tommi}, \cite{Montiel:2013vn}. In cooperative game theory \cite{montiel_c}, MCMC can be used to create an approximate objective function to optimize the negotiation strategy for a coalition of players.

%%%%%%%%%%%%%%%%%%%%%%%%%%%%%%%%%%%%%%%%%%%%%%%%%%%%%%%%%%%%%%%%%%%%
 
\subsection{The blueprint}

This work presents an algorithm we call Matrix Hit and Run (MHAR) for sampling full and non-full dimensional polytopes. MHAR enhances the Hit-and-Run (HAR) algorithm proposed in \cite{montiel_jd}. We use the standard definition of a generic polytope $\Delta:=\{x \in \R^n | Ax \leq b\}$, where $(A,b) \in \R^{m \times n} \times \R^{m \times 1}$, $n$ is the number of elements of $x$, $m = m_E + m_I$ is the number of restrictions, $m_E$ is the number of equality constraints, and $m_I$ is the number of inequality constraints. 

The contribution of this work is six-fold:
\begin{itemize}
\item First, we introduce Matrix Hit-and-Run (\textbf{MHAR}).
\item Second, we show that the cost per sample of the MHAR depends entirely on $m, n, z$, and $\omega$, where $m, n$ are as described in the definition of $\Delta$, $\omega$ represents a matrix multiplication coefficient as described in Table \ref{table:1}, and $z$ is a padding hyper-parameter specified by the user. After proper pre-processing and a warm start, the algorithm has a cost per sample of $\SO\big(\min(m_I^{\omega -2}n^4, m_In^{\omega +1})\big)$ for the full dimensional scenario, and of $\SO\big(\min(n^{\omega +2}, m_In^{\omega +1})\big)$ for the non-full dimensional one.
\item Third, we demonstrate that MHAR has lower cost per sample than HAR if the hyper-parameter $z$ is bigger than $\max(n,m)$. This is achieved by switching possibly isolated \textit{walks} into a padded matrix that allows us to share operations between walks. 
\item Fourth, we show that after proper pre-processing and a warm start, MHAR has a lower asymptotic cost per sample for the regime $n^{1+\frac{1}{3}} \ll m$ than does any of the published sampling algorithms \cite{walks}.   
\item Fifth, we provide code for MHAR as a python library based on the Pytorch framework. It is ready for use in CPU or CPU-GPU architectures (as found in Colab, AWS, Azure, and Google Cloud). All MHAR experiments were conducted using Colab notebooks with an Nvidia P100 GPU. The code is available in \url{https://github.com/uumami/mhar_pytorch}. The python package can be installed with the \textbf{pip install mhar}, the official site of the package is \url{https://github.com/uumami/mhar}
\item Sixth, we present the results of experiments to assess the performance of MHAR against the \textit{hitandrun} package used in \cite{har_tommi}. MHAR was found to be substantially faster in almost all scenarios, especially in high dimensions. Furthermore, we ran simulations to empirically test the convergence in distribution of our implementation, with favorable results. Finally we present insights over the padding hyper-parameter $z$ obtained via computational tests.
\end{itemize}

The remainder of this paper is organized as follows. \S \ref{S:2} revises definitions and some basic matrix-to-matrix operations. \S \ref{S:3} revisits the cost per iteration and cost per sample of HAR. \S \ref{S:4} provides a computational complexity analysis of MHAR. \S \ref{S:5} compares MHAR against other algorithms developed for full dimensional scenarios. \S \ref{S:6} contains a back-to-back comparison of our implementation against the \textit{``hitandrun''} library used in \cite{har_tommi}, and a numerical analysis of the padding parameter $z$. \S \ref{S:7} presents our conclusions and identifies future work.

For clarity and simplicity, HAR will refer to the algorithm presented in \cite{montiel_jd}, which extends \cite{firs_har_ever} for non-full dimensional polytopes. For ease of comparison, we use "soft-O" notation $\SO$, which suppresses $log(n)$ factors and other parameters like error bounds \cite{walks}, \cite{RHMC}, \cite{original-har}. In order to allow comparison with other algorithms, we assume that the polytope sampled by HAR and MHAR has received proper pre-processing, which means the polytope is in near isotropic position as defined in \cite{walks}, \cite{RHMC}, \cite{har_tommi}. Additionally all algorithms are compared from a \textit{warm start}. We use $f \ll g$ notation to define a relation where $f\in \OO(g)$. Finally, we assume the existence of a random stream of bits that allow us to generate a random number in $\OO(1)$.

%%%%%%%%%%%%%%%%%%%%%%%%%%%%%%%%%%%%%%%%%%%%%%%%%%%%%%%%%%%%%%%%%%%%
%%%%%%%%%%%%%%%%%%%%%%%%%%%%%%%%%%%%%%%%%%%%%%%%%%%%%%%%%%%%%%%%%%%%
%%%%%%%%%%%%%%%%%%%%%%%%%%%%%%%%%%%%%%%%%%%%%%%%%%%%%%%%%%%%%%%%%%%%
%%%%%%%%%%%%%%%%%%%%%%%%%%%%%%%%%%%%%%%%%%%%%%%%%%%%%%%%%%%%%%%%%%%%

\section{Preliminaries}\label{S:2}
This section formalizes the notation and provides a brief overview of computational complexity in matrix-to-matrix operations.

%%%%%%%%%%%%%%%%%%%%%%%%%%%%%%%%%%%%%%%%%%%%%%%%%%%%%%%%%%%%%%%%%%%%
\subsection{Polytopes}

We start by defining a polytope, which is the n-dimensional generalization of a polyhedron, as the intersection of half-spaces. Formally, a polytope is characterized by a set of $m_E$ linear equality constraints and $m_I$ linear inequality constraints in a Euclidean space ($\mathbb{R}^n$):
\begin{align}
 \Delta^I \ &= \{ x\in \mathbb{R}^n \;|\;  A^Ix \leq b^I, \;  A^I \in \mathbb{R}^{m_I \times n}, \;  b^I \in \mathbb{R}^{m_I} \}, \label{TS1}\\
 \Delta^E &= \{ x\in \mathbb{R}^n \;|\;  A^Ex = b^E, \;  A^E \in \mathbb{R}^{m_E \times n}, \;  b^E \in \mathbb{R}^{m_E} \},\label{TS2}\\
 \Delta \ \ &= \Delta^I \cap \Delta^D, \label{TS3}
\end{align}
where Equations (\ref{TS1}) and (\ref{TS2}) are defined by the inequalities and equalities, respectively. The third equation defines the polytope of interest, and  it is the intersection of the two previous sets. Since $\Delta$ is the intersection of convex sets, then by construction it is also convex. For simplicity we assume all polytopes to be bounded, non-empty, and characterized with no redundant constraints.

%%%%%%%%%%%%%%%%%%%%%%%%%%%%%%%%%%%%%%%%%%%%%%%%%%%%%%%%%%%%%%%%%%%%
\subsection{Matrix multiplication}

We adopt common notation used in matrix multiplication. $\omega$ represents the matrix multiplication coefficient - which characterizes the number of operations required to multiply two $n \times n$ matrices. The complexity for such multiplication is of the order $\OO(n^\omega)$. The lowest complexity for matrix multiplication algorithms is conjectured to be $\Omega(n^2)$ \cite{lowest_bound}. Table \ref{table:1} shows the theoretical bounds for many well-known multiplication algorithms.
\begin{table}[h!]
\caption{Asymptotic complexity of matrix multiplication algorithms}
\label{table:1}
\vspace{.3cm}
\centering
\begin{tabular}{ |p{4.5cm}||p{2.5cm}|  }
 \hline
 \multicolumn{2}{|c|}{Matrix Multiplication Algorithms} \\
 \hline
 \textbf{Algorithm} & \textbf{Complexity}\\
 \hline
 Naive   & $\OO(n^{3})$   \\
 \hline
 Strassen-Schonhaeg &  $\OO(n^{2.807})$ \\
 \hline
 Coppersmith-Winograd & $\OO(n^{2.376})$ \\
 \hline
 Legall     & $\OO(n^{2.373})$ \\
 \hline
\end{tabular} 
\end{table}

In general, \cite{rec_matrix} showed that the number of operations needed to multiply two matrices with dimensions $m\times n$ and $n \times p$ is of $\mathcal{O}(d_1d_2d_3^{\omega-2})$, where $d_3=min\{m,n,p\}$ and $\{d_1,d_2\} = \{m,n,p\}-\{d_3\}$. The special case of matrix-vector multiplication $d_3=1$ yields a bound of $\mathcal{O}(mn)$. The smallest published $\omega$ is 2.373 \citep{legall_min}.

It is possible to define a function $\mu$ that represents the matrix multiplication order of complexity for matrices $A \in \R^{n_1 \times n_2}$ and $B \in \R^{n_2 \times n_3}$ as
\begin{align}
\mu_{A,B}=&
  \begin{cases}
    n_1^{\omega-2} n_2 n_3       & \quad \text{if } \min\{n_1,n_2,n_3\}=n_1,\\
    n_1 n_2^{\omega-2} n_3       & \quad \text{if } \min\{n_1,n_2,n_3\}=n_2,\\
    n_1 n_2 n_3^{\omega-2}       & \quad \text{if } \min\{n_1,n_2,n_3\}=n_3.
  \end{cases}
 \end{align}
Thus we can express the complexity of the operation $AB$ as $\OO(\mu_{A,B})$.

In practice, only the Naive and Strassen's algorithms are used because the constants hidden in the Big O notation are usually significantly big for large enough matrices to take advantage of. Moreover, many multiplication algorithms are impractical due to numerical instabilities \cite{walks}. Fortunately, there have recently been fast and numerically stable implementations of the Strassen algorithm using GPUs (\cite{strassen_gpu}, \cite{recipes}, \cite{reliable_gpu}).

%%%%%%%%%%%%%%%%%%%%%%%%%%%%%%%%%%%%%%%%%%%%%%%%%%%%%%%%%%%%%%%%%%%%
%%%%%%%%%%%%%%%%%%%%%%%%%%%%%%%%%%%%%%%%%%%%%%%%%%%%%%%%%%%%%%%%%%%%
%%%%%%%%%%%%%%%%%%%%%%%%%%%%%%%%%%%%%%%%%%%%%%%%%%%%%%%%%%%%%%%%%%%%
%%%%%%%%%%%%%%%%%%%%%%%%%%%%%%%%%%%%%%%%%%%%%%%%%%%%%%%%%%%%%%%%%%%%%%%%%%%%%%%

\section{HAR} \label{S:3}
This section explains the HAR algorithm and calculates its cost per iteration and mixing time for non-full dimensional polytopes, as defined in \cite{montiel_jd}.

%%%%%%%%%%%%%%%%%%%%%%%%%%%%%%%%%%%%%%%%%%%%%%%%%%%%%%%%%%%%%%%%%%%%
 \subsection{Overview}
HAR can be described as follows. A \textit{walk} is initialized in a strict inner point of the polytope. At any iteration, a random direction is generated via independent normal variates. The random direction, along with the current point, generates a line set $L$, and its intersection with the polytope generates a line segment. The sampler selects a random point in $L$ and repeats the process. After a warm start, HAR for full-dimensional convex bodies has a cost per iteration $\OO(m_In)$ and a cost per sample of $\SO(m_I n^4)$ \cite{walks}.

In general, the non-deterministic mixing time of HAR is of $\SO(n^2 \gamma_{\kappa})$, where $ \gamma_{\kappa}$ is defined as
\[ 
 \gamma_{\kappa} = \inf_{R_{in}, R_{out}>0}\Bigg\{ \frac{R_{out}}{R_{in}} \| \mathcal{B}(x, R_{in}) \subseteq \Delta \subseteq \mathcal{B}(y, R_{out}) \ for \ some \ x,y \in \Delta\Bigg\},
\]
where $R_{in}$ and $R_{out}$ are the radii of an inscribed and circumscribed ball of the polytope $\Delta$, respectively, and $\mathcal{B}(q,R)$ is the ball of radius $R$ containing the point q. In essence,
$ \gamma_{\kappa}$ is the coefficient generated by the biggest inscribed ball and the smallest circumscribed ball of the polytope. That the mixing time depends on these parameters means that elongated polytopes are harder to sample. Implementations of HAR for convex bodies are typically analyzed after  pre-processing and invoking a warm start, meaning that the body in question is brought to a near isotropic position in $\SO(\sqrt{n})$, allowing the mixing time to be expressed as $\SO(n^3)$ \cite{walks}, \cite{geo-walk}, \cite{RHMC}, \cite{ballwalk}, \cite{original-har} and \cite{har_tommi}. For ease of comparison with the literature, the remainder of the paper assumes that the polytope has received proper pre-processing.
 
A HAR sampler must compute the starting point and find the line segment $L$ at each iteration. Additionally, a thinning factor (also called "burning rate") $\varphi(n)$ must be included to achieve a fair almost uniform distribution over the studied space \cite{har_tommi}. This means that after a warm start, the algorithm needs to drop $\varphi(n)$ sampled points for each desired i.i.d. observation. This thinning factor is known as the mixing-time, which is $\SO(n^3)$ in the case of polytopes (see \cite{walks}, \cite{har_tommi}, \cite{original-har}).

The HAR pseudocode proposed in \cite{montiel_jd} for full and non-full dimensional polytopes is presented in Algorithm \ref{alg:har}. It samples a collection $\mathcal{X}$ of $T$ uncorrelated points inside $\Delta$. We know the complexity of HAR for full dimensional polytopes, to find the cost per iteration and cost per sample of HAR for non-full dimensional polytopes will require analyzing the complexity of calculating the projection matrix. 

%%%%%%%%%%%%%%%%%%%%%%%%%%%%%%%%%%%%%%%%%%%%%%%%%%%%%%%%%%%%%%%%%%%%
 
\subsection{Projection matrix}

The projection matrix $P_{\Delta^E}$ is computed from the equality matrix $ A^E $. Then, $P_{\Delta^E}$ allows any vector to be projected to the null space of $A^E$. In our case, the random direction vector $h$ lives in a full dimensional space, which means that if $m_E>0$, $h$ needs to be projected so that the line set $L$ lives in the same space as $\Delta$. The projection operation $P_{\Delta^E}h = d$ yields $A^Ed=0$. Then, $A^E(x+d)=b^E$ \cite{montiel_jd}. 
(This step is omitted if $m_E=0$.) 

The projection matrix is defined as
\begin{equation}
P_{\Delta^E} = I  - A'^E(A^EA'^E)^{-1}A^E .
\end{equation}

\begin{algorithm}[t] \label{alg:har}
 \caption{HAR pseudocode}
\SetAlgoLined
\KwResult{$\mathcal{X}$}
Initialization\;
    $t\gets 0$  (Sample point counter)\;
    $j\gets 0$ (Iteration counter)\;
    $\mathcal{X} =\emptyset$\;
    \textbf{Set} the total sample size $T$\;    
    \textbf{Set} a thinning factor $\varphi(n)$\;
    \textbf{Find} a strictly inner point of the polytope $\Delta$ and label it $x_{t=0,j=0}$\;
    \If{$(m_E>0)$}{Compute the projection matrix $P_{\Delta^E}$}
        \While{$t < T$}{
                \textbf{Generate} the direction vector $h \in \mathbb{R}^n$\;
                \eIf{$(m_E=0)$}{$d=h$}{$d=P_{\Delta^E}h$}
                \textbf{Find} the line set $L:=\{x|x=x_{t,j}+\theta d, x\in\Delta  \; \& \; \theta \in \mathbb{R} \}$\;
                $ j\gets j+ 1 $\;
                \textbf{Generate} a point uniformly distributed in $L \cap \Delta$ and label it $x_{t,j+1}$\;
                \If{$j==\varphi(n)$}{
                    $ \mathcal{X} = \mathcal{X} \cup x_{t,j} $\;
                    $ t\gets t + 1 $\;
                    $ j\gets 0 $\;
                }
 }
\end{algorithm}

\begin{lemma}
\label{lemma projection complexity}
If $m_E < n$, then the complexity of calculating $P_{\Delta^E}$ is $\OO(m_E^{\omega-2}n^2)$.
 \end{lemma}
 \begin{proof}
Computing $P_{\Delta^E}$ is done in three matrix multiplications, one matrix-to-matrix subtraction, and one matrix inversion operation over $(A^EA'^E)$. The number of operations needed to calculate the inverse matrix depends on the algorithm used for matrix multiplication \cite{inversion}. The order of number of operations for computing $P_{\Delta^E}$ is the sum of the following: 
\begin{enumerate}
  \item Obtain $(A^EA'^E)$ in $\mathcal{O}(\mu_{A^E, A'^E})=\mathcal{O}(\mu(m_E,n,m_E)) =\mathcal{O}(m_E^{\omega-1}n)$ operations.
  \item Find the inverse $(A^EA'^E)^{-1}$ in $\mathcal{O}(m_E^{\omega})$, since  $(A^EA'^E)^{-1}$ has dimension $m_E \times m_E$.
  \item Multiply $ A'^E(A^EA'^E)^{-1}$ in $\mathcal{O}(\mu_{ A'^E,(A^EA'^E)^{-1}})=\mathcal{O}(\mu(n,m_E,m_E)) =\mathcal{O}(m_E^{\omega-1}n)$.
  \item Calculate $A'^E(A^EA'^E)^{-1}A^E$ in $\mathcal{O}(\mu_{ A'^E(A^EA'^E)^{-1},A^E})=\mathcal{O}(\mu(n,m_E,n)) =\mathcal{O}(m_E^{\omega-2}n^2)$.
  \item Subtract $ I  - A'^E(A^EA'^E)^{-1}A^E$ in $\OO(n^2)$. 
\end{enumerate}

These sum to $2 \times \mathcal{O}(m_E^{\omega-1}n) + \mathcal{O}(m_E^{\omega}) + \OO(m_E^{\omega-2}n^2)+\mathcal{O}(n^2)$. Hence the complexity of calculating $P_{\Delta^E}$ is $ \OO(\mu_{A'^E(A^EA'^E)^{-1},A^E})=\OO(m_E^{\omega-2}n^2)$.
\end{proof}

For simplicity, we will denote the complexity of computing $P_{\Delta^E}$ as $\OO(\mu_{P_{\Delta^E}})$.

\subsection{Non-full dimensional HAR}
We proceed to calculate the cost per sample of HAR for $m_E > 0$. We start by computing the cost per iteration in Lemma \ref{lemma cost per iteration har nf}.

\begin{lemma}\label{lemma cost per iteration har nf} 
The \textit{cost per iteration} of HAR for $0 \leq m_E$ is $\OO(\max\{ m_In,m_E^{\omega-2}n^2\})$.
\end{lemma}
\begin{proof}
As seen in Algorithm \ref{alg:har}, the only difference between the full and non-full dimensional cases is the projection step $P_{\Delta^E}h=d$. Then, the cost per iteration is defined by the larger of the original cost per iteration $\OO(m_In)$ of HAR for $m_E=0$, and the extra cost induced by the projection when $m_E>0$.

Because $P_{\Delta^E}$ has dimension $n \times n $ and $h$ is an $n \times 1$ vector, $\mu_{P_{\Delta^E},h}=n^2$ and the complexity is $\OO(n^2)$. By Lemma \ref{lemma projection complexity}, finding $P_{\Delta^E}$ has an asymptotic complexity of $\OO(m_E^{\omega-2}n^2)$. Therefore, the cost of projecting $h$ at each iteration is $ \OO(n^2) + \OO(m_E^{\omega-2}n^2) = \OO(m_E^{\omega-2}n^2)$, since $m_E>0$. Therefore, the cost per iteration for $m_E >0$ is $\OO(\max\{ m_In,m_E^{\omega-2}n^2)\})$. 
If $m_E=0$, then the coefficient $\max\{ m_In,m_E^{\omega-2}n^2)\}$ equals $\max\{m_In,0\}=m_In$ and the cost per sample is $\SO(\max\{m_In,0)\})=\SO(m_In)$.
\end{proof}

Having calculated the cost per iteration of HAR, we can proceed to Theorem \ref{th_cost_per_sample_har_nf}.
\begin{theorem}\label{th_cost_per_sample_har_nf}
The \textit{cost per sample} of HAR for $\ 0 \leq m_E$ is $\SO(n^3\max\{ m_In,m_E^{\omega-2}n^2\})$ after proper pre-processing and a warm start.
\begin{proof}
According to \cite{walks}, the cost per sample of a sampling algorithm is its mixing time complexity multiplied by its cost per iteration. By Lemma \ref{lemma cost per iteration har nf}, the cost per iteration is $\OO(\max\{ m_In,m_E^{\omega-2}n^2\})$. Moreover, \cite{har_tommi} states that the mixing time, after a warm start, of HAR is $\SO(n^3)$. Therefore, the cost per sample is $\SO(n^3\max\{ m_In,m_E^{\omega-2}n^2\})$. 

Recall that if $m_E=0$ the \textit{cost per sample} is $\SO(n^3\max\{m_In,0)\})=\SO(m_In^4)$ that is the special case of HAR for full dimensional polytopes. 
\end{proof}
\end{theorem}

%%%%%%%%%%%%%%%%%%%%%%%%%%%%%%%%%%%%%%%%%%%%%%%%%%%%%%%%%%%%%%%%%%%%
%%%%%%%%%%%%%%%%%%%%%%%%%%%%%%%%%%%%%%%%%%%%%%%%%%%%%%%%%%%%%%%%%%%%
%%%%%%%%%%%%%%%%%%%%%%%%%%%%%%%%%%%%%%%%%%%%%%%%%%%%%%%%%%%%%%%%%%%%
%%%%%%%%%%%%%%%%%%%%%%%%%%%%%%%%%%%%%%%%%%%%%%%%%%%%%%%%%%%%%%%%%%%%

\section{Matrix Hit-and-Run (MHAR)}\label{S:4}

This section details our new algorithm, Matrix Hit-And-Run (\textbf{MHAR}). MHAR has a lower cost per sample than does HAR. Furthermore, making $z$ simultaneous walks with MHAR requires fewer operations than does running $z$ HAR walks in parallel. The "padding" hyper-parameter $z$ allows the concatenation of multiple directions $d$ and samples $x$ to form matrices $D$ and $\mathcal{X}$, respectively. Each column of these matrices represents a walk over the polytope. This modification permits the use of efficient matrix-to-matrix operations to simultaneously project many directions $d$ and find their respective line segments.

%%%%%%%%%%%%%%%%%%%%%%%%%%%%%%%%%%%%%%%%%%%%%%%%%%%%%%%%%%%%%%%%%%%%
 
\subsection{MHAR preliminaries}
MHAR explores the polytope using simultaneous walks by drawing multiple directions $d$ from the n-dimensional hypersphere. Each independent walk has the same mixing-time as with HAR, but a lower cost per iteration. Instead of running separate threads, we "batch" the walks by "padding" vector $x$ and $d$ with $z$ columns, creating the matrices $X=(x^1|\dots|x^k|\dots|x^z)$ and $D=(d^1|\dots|d^k|\dots|d^z)$. Super index $k$ denotes the $k$th walk represented by the $k$th column in the padded matrix. The algorithm then adapts the steps in HAR to keep track of each independent walk and recast the operations as matrix-to-matrix. The algorithm is tailored for exploiting cutting-edge matrix routines that exploit the architectures of machines like GPUs, cache memories, and multiple cores.

The main difference with HAR when running $z$ instances on multiple independent cores ($z$-HAR) is the estimation of $D=(d^1|\dots|d^k|\dots|d^z)$ and the line segments $L^k$ in a simultaneous fashion for all $z$-walks. In both, $z$-HAR and MHAR, each walk is oblivious of the others after a warm start, which guarantee a constant mixing-time among all $z$-walks \cite{montiel_jd} \cite{starting_point}. 

Algorithm \ref{alg:mhar} presents the pseudocode for MHAR.  
 %%%%%%%%%%%%%%%%%%%%%%%%%%%%%%%%%%%%%%%%%%%%%%%%%%%%%%%%%%%%%%%%%%%%
 
%%% MHAR ALGORITHM
\begin{algorithm}[h]
\label{alg:mhar}
\caption{MHAR pseudocode}
\KwResult{$\mathcal{X}$}
 Initialization\;
    $t\gets 0$  (Sample point counter)\;
    $j\gets 0$ (Iteration counter)\;
    $z \gets \max\{m_I,n\} + 1$\;
    $\mathcal{X} = \emptyset$\;
    \textbf{Set} the total sample size $T$\;    
    \textbf{Set} a thinning factor $\varphi(n)$\;
    \textbf{Find} a strictly inner point of the polytope $\Delta$ and label it $x_{t,j}$\;
    \textbf{Set} $x_{t,j}^k=x_{t,j}, \ \forall k \in \{1,...,z\}$\;
    \textbf{Initialize} $X_{t,j}=(x_{t,j}^1|...|x_{t,j}^k|...|x_{t,j}^z) \in \R^{n \times z}$\;
    \If{$(m_E>0)$}{Compute the projection matrix $P_{\Delta^E}$}
        \While{$t < T$}{
                \textbf{Generate} $H=(h^1|...|h^k|...|h^z) \in \mathbb{R}^{n \times z}$, the direction matrix\;
                \eIf{$(m_E=0)$}{$D=H$;}{$D=P_{\Delta^E}H=(d^1|...|d^k|...|d^z)$;}
                \textbf{Find} the line sets $\Big\{L^k:=\{x|x=x_{t,j}^k+\theta^k d^k, \ x\in\Delta  \; \& \; \theta^k \in \mathbb{R} \}\Big\}_{k=1}^z$\;
                $ j\gets j + 1 $\;
                \textbf{Generate} a point uniformly distributed in each $L^k$ and label it $x_{t,j}^k$ in $X_{t,j}$\;
                \If{$j==\varphi(n)$}{
                    $ \mathcal{X} = \mathcal{X} \cup \{x^1_{t,j}, ..., x^z_{t,j}\} $\;
                    $ t\gets t + z $\;
                    $ j\gets 0 $\;
                }
 }
\end{algorithm}

%%%%%%%%%%%%%%%%%%%%%%%%%%%%%%%%%%%%%%%%%%%%%%%%%%%%%%%%%%%%%%%%%%%%
%%%%%%%%%%%%%%%%%%%%%%%%%%%%%%%%%%%%%%%%%%%%%%%%%%%%%%%%%%%%%%%%%%%%
%%%%%%%%%%%%%%%%%%%%%%%%%%%%%%%%%%%%%%%%%%%%%%%%%%%%%%%%%%%%%%%%%%%%
\subsection{Starting point}
In general, the cost of finding the starting point is excluded from the complexity analysis because it is independent of the mixing-time. However, we present it here for completeness even though the literature assumes a warm start in determining cost per sample (\cite{walks}, \cite{har_tommi}, \cite{RHMC}).

MHAR needs to be initialized by a point in the relative interior of the polytope. We suggest Chebyshev's center of the polytope, which is the center of the largest inscribed ball. For polytopes, Chebyshev's center can be formulated as a linear optimization problem and solved using standard methods. 

Chebyshev's center is presented in Model (\ref{CCHC}).
\begin{equation}\label{CCHC}
\begin{aligned}
 \max \limits_{{x\in \R^n, r\in\R}} \ &r,\\
s.t \quad A^Ex&=b^E,\\ 
 (a^I_i)^Tx + r||a^I_i||_2 &\leq b_i^I, \ \forall i = 1,...,m_I,
\end{aligned}
\end{equation}
where $a^I_i$ and $b^I_i$ represent the $i$th row of matrix $A^I$ and $i$th entry from vector $b^I$, respectively. Model (\ref{CCHC}) has the original $m$ restrictions plus one additional variable $r$. Hence, the size of the problem has $m$ constraints and $n+1$ variables. Then, calculating the $||\cdot||_2$ coefficients takes $\OO(mn)$. Thus, it can be formulated and solved in $\mathcal{O}(n^{\omega})$ using Vaidya's algorithm \cite{fast_linear} for linear optimization. After solving Model (\ref{CCHC}), we use $x$ as the starting point $x_{t=0, j=0}$ for all walks and draw independent walking directions. 
The matrix $X_{t,j} \in \R^{n \times z}$ introduced in Algorithm \ref{alg:mhar} is the algorithmic version of  $X$, and it summarizes the state of all walks, where each $k$th column represents the current point of walk $k$ at iteration $\{t,j\}$. Formally we say $X_{t,j} = (x_{t,j}^1|...|x_{t,j}^k|...|x_{t,j}^z)$ where $x_{t,j}^k \in \R^{n \times 1} \ \forall k \in \{1,...,z\}$.

%%%%%%%%%%%%%%%%%%%%%%%%%%%%%%%%%%%%%%%%%%%%%%%%%%%%%%%%%%%%%%%%%%%%
\subsection{Generating D}

Because the target distribution of HAR and MHAR is uniform, we follow the procedure established in \cite{montiel_jd} and \cite{original-har} that uses the Margsalia method \cite{mar} to generate a random vector $h$ from the hypersphere by generating $n$ i.i.d. samples from a standard normal distribution $\mathcal{N}(0, \ 1)$. However, instead of generating a single direction vector $d \in \R^ n$, we create matrices $H,\ D \in \R^{n \times z}$, where each element of the matrix corresponds to an independent execution of the Box-Muller method \cite{box-muller} bounded by $\OO(nz)$. If the polytope is full dimensional, $H=D$ and no projection operation is needed. Otherwise, the projection matrix $P_{\Delta^E}$ is calculated as in \S \ref{S:3}, and Lemma \ref{lemma projection complexity} bounds the number of operations as $\OO(m_E^{\omega-2}n^2)$. 

Matrices $H$ and $D$ can be visualized as
\begin{align}
H = (h^1 | ... | h^k | ... | h^z),\ h^k \in \R^n,\ \forall k \in \{1,...,z\}, \\
D = P_{\Delta^E}H = (d^1 | ... | d^k | ... | d^z),\ d^k \in \R^n,\ \forall k \in \{1,...,z\}.
\end{align}
Each column $h^k$ can be projected by the operation $D=P_{\Delta^E}H$. Hence, each column of $D$ satisfies the restrictions in $\Delta^E$ and serves as a direction $d$ for an arbitrary walk $k$. In principle, $z$ can be any number in $\N$, where $z=1$ is the special case that recovers the original HAR.

%%%%% Lemmma generating D
\begin{lemma}\label{lemma complexity generating D}
The complexity of generating matrix D in MHAR given $P_{\Delta^E}$ and $\max\{m_I, n\} \leq z$ is $\OO(nz)$ if $m_E=0$, and $\OO(n^{\omega-1}z)$ if $m_E>0$.
\begin{proof}
Generating $H$ has complexity $\OO(nz)$ using the Box-Muller method. If $m_E=0$, then $D=H$, implying a total asymptotic cost $\OO(nz)$. If $m_E>0$, then $D=P_{\Delta^E}H$, whose cost $\OO(\mu_{P_{\Delta^E},H})=\OO(n^{\omega-1}z)$ given by $\max\{m_I, n\} \leq z$, needs to be included. $\OO(n^{\omega-1}z)$ bounds $\OO(nz)$. Therefore, the total cost of computing $D$ for $m_E>0$ is bounded by $\OO(n^{\omega-1}z)$.
\end{proof}
\end{lemma}

Lemma \ref{lemma complexity generating D} shows that if $m_E>0$, the cost of generating new directions $d$ does not scale as if had used $z$ parallel HARs. In the HAR case, the operations required would have been carried out in $\OO(z \mu_{P_{\Delta^E},h})=\OO(zn^2)$, averaging $\OO(\frac{zn^2}{z})=\OO(n^2)$ per direction. In contrast, MHAR is $\OO(n^{\omega-1}z)$, averaging  $\OO(\frac{n^{\omega-1}z}{z})=\OO(n^{\omega-1})$ per direction. When $m_E=0$, the number of operations for both cases is the same.

%%%%%%%%%%%%%%%%%%%%%%%%%%%%%%%%%%%%%%%%%%%%%%%%%%%%%%%%%%%%%%%%%%%%
\subsection{Finding the line sets}
Given matrices $X$ and $D$, we now obtain the line sets $\{L^k\}_{k=1}^{z}$:
\begin{equation}
\Big\{L^k:=\{x|x=x^k+\theta^k d^k, \ x\in\Delta,  \; \mbox{and} \; \theta^k \in \mathbb{R} \}\Big\}_{k=1}^z.
\end{equation}

Each $\theta^k$ characterizes the line set for column $x^k$. The "padded" column-wise representation of restrictions $\Delta^I$ is
\begin{equation}
A^{I}X=\begin{pmatrix}
a_1^Ix^1 & \dots &  a_1^Ix^k \\ 
\vdots & \ddots & \vdots\\ 
a_{m_I}^Ix^1 & \dots &  a_{m_I}^Ix^k
\end{pmatrix}
\leq 
\begin{pmatrix}b_1^I \\ \vdots \\  b_{m_I}^I   \end{pmatrix}=b^{I},
\end{equation}
where each element from the left matrix must be less than or equal to the corresponding element (row-wise) in vector $b^I$.
The restrictions for an arbitrary $x^k$ can be rewritten row-wise so that the left side and right side are scalars:
\begin{align} 
a_i^Ix^k \leq b^I_i, \ \forall i \in \{1, \dots, m_I \}.
\end{align}
Then, each $\theta^k$s must satisfy
\begin{align} 
(a_i^Ix^k + \theta^k a_i^Id^k) < b^I_i, \ \forall i \in \{1, \dots, m_I \}.
\end{align}

Rearranging the terms obtains restrictions for each walk $k$, where each $\theta^k$ must be bounded by its respective set of lambdas $\{\lambda^k_i\}_{i=1}^{m_I}$, as follows:
\begin{align}
\theta^k &< \lambda_i^k = \frac{b_i^I - a_i^Ix^k}{a_i^I d^k}, \quad if \ a_i^I d^k>0, \\
\theta^k &> \lambda_i^k = \frac{b_i^I - a_i^Ix^k}{a_i^I d^k}, \quad if \ a_i^I d^k<0.
\end{align}
Hence, a walk's boundaries are represented by
\begin{align}
 \lambda_{min}^k=\max\ \{\lambda_i^k \ | \ a_i^I d^k < 0\},   \\
 \lambda_{max}^k=\min\ \{\lambda_i^k \ | \ a_i^I d^k > 0\}.  
\end{align}
These lambdas can be used to construct the intervals $\Lambda^k=(\lambda^k_{\min},$ $\lambda^k_{\max}), \ k\in \{1,...,z\}$. 
By construction, if $\theta^k \in \Lambda^k$ and $x^k \in \Delta$, then $x^k + \theta^kd^k \in L^k$, since $A^I(x^k + \theta^k d^k) \leq b^I $ and $A^E(x^k + \theta^k d^k)=b^E$. The line segment can be found simply by evaluating $\{\Lambda^k\}_{k=1}^{z}$, because $x^k$ and $D$ were computed previously. We can now state Lemma \ref{lemma complexity generating L}.

%%%%%%%%%%%%%%%%%%%%%%%%%%%%%
\begin{lemma}\label{lemma complexity generating L}
The complexity of generating all line sets $\{L^k\}_{k=1}^{z}$ in MHAR given $D$, $X$, and $\max\{m_I, n\} \leq z$ is bounded by $\OO(m_In^{\omega -2}z) \ if \ n \leq m_I$, and by $\OO(m_I^{\omega -2}nz)$ otherwise.
\begin{proof}
All $\Lambda^k$s can be obtained as follows:
\begin{enumerate}
    \item Obtain matrix $A^IX$ in $\OO(\mu_{A^I,X})$. This is done in $\OO(m_In^{\omega -2}z) \ if \ n \leq m_I$, and in $\OO(m_I^{\omega -2}nz)$ otherwise.
    \item Compute $B^I - A^IX$, where $B_I=(b^I|...|b^I) \in \R^{m^I \times z}$, which takes $\OO(m_Iz)$ operations.
    \item Calculate $A^ID$, which is bounded by  $\OO(\mu_{A^I,D})$, which is done in $\OO(m_In^{\omega -2}z) \ if \ n \leq m_I$, and in $\OO(m_I^{\omega -2}nz)$ otherwise.
    \item Divide $\frac{B^I - A^IX}{A^ID}$ (entry-wise) to obtain all $\lambda^k_i$. All the necessary point-wise operations for this calculation have a combined order of $\OO(m_Iz)$. 
    \item For each $k \in \{1,...,z\}$, find which coefficients $a_i^I d^k$ are positive or negative, which takes $\OO(m_Iz)$.
    \item For each $k \in \{1,...,z\}$, find the intervals $\lambda_{min}^k=\max\ \{\lambda_i^k \ | \ a_i^I d^k < 0\}$ and $ \lambda_{max}^k=\min\ \{\lambda_i^k \ | \ a_i^I d^k > 0\}$, which can be done in $\OO(m_Iz)$.
\end{enumerate} 

This procedure constructs all the intervals $\Lambda^k=(\lambda_{\min}^k, \lambda_{\max}^k)$. The complexity of this operation is bounded by $\OO(\mu_{A^I,X}) = \OO(\mu_{A^I,D})$. Hence, the complexity of finding all line sets is bounded by $\OO(m_In^{\omega -2}z) \ if \ n \leq m_I$, and by $\OO(m_I^{\omega -2}nz)$ otherwise. 
\end{proof}
\end{lemma}
Lemma \ref{lemma complexity generating L} bounds the complexity of finding the line sets at any iteration of MHAR. This leaves only analyzing the cost of choosing a new sample.

%%%%%%%%%%%%%%%%%%%%%%%%%%%%%%%%%%%%%%%%%%%%%%%%%%%%%%%%%%%%%%%%%%%%
\subsection{Choosing samples}
The following lemma bounds the complexity of choosing a new $X_{t,j+1}$ or $X_{t+z,0}$ given $\Lambda^k \ \forall k\in \{1,...,z\}$. The new samples will be padded to create the matrix $X_{t,j+1}=(x_{t,j+1}^1 | \dots | x_{t,j+1}^k)$ to be used in the next iteration.
\begin{lemma}\label{lemma complexity sample pick}
Sampling $z$ new points given $\{\Lambda^k\}_{k=1}^z$ has complexity $\OO(zn)$.
\begin{proof}
Selecting a random $\theta^k \in \Lambda^k$ takes $\OO(1)$. Sampling a new point $x^k_{t,j+1} = x^k_{t,j} + \theta d^k_{t,j}$ has complexity $\OO(n)$ because it requires $n$ scalar multiplications and $n$ sums. Then, sampling all new $x_{t,j+1}^k$ points is bounded by $\OO(zn)$.
\end{proof}
\end{lemma}
Having concluded the complexity analysis for each step of the loop, we next calculate the cost per iteration and proceed to measure the cost per sample.

%%%%%%%%%%%%%%%%%%%%%%%%%%%%%%%%%%%%%%%%%%%%%%%%%%%%%%%%%%%%%%%%%%%%
\subsection{Iteration and sampling costs of MHAR}

The asymptotic behavior of each operation that comprises the main loop of MHAR when $\max\{n,m_I\} \leq z$ is presented in Table \ref{complexity_mhar}. The cost of finding the starting point is excluded (\cite{walks}, \cite{har_tommi}).
\vspace{-.2cm}

\begin{table}[h!]
\setstretch{1.5} 
\caption{Asymptotic cost per sample of MHAR at each step}
\label{complexity_mhar}
\vspace{.2cm}
\centering
\begin{tabular}{ |p{3.7cm}||p{2.2cm}||p{2.2cm}||p{2.2cm}||p{2.2cm}|  }
 \hline
 \multicolumn{5}{|c|}{MHAR complexity at each step, $(n,m) < z$} \\
\hline
Operation & $m_E=0, \quad \ $ $\quad n \leq m_I.$ & $m_E=0, \quad \ $ $ \ n>m_I.$ &  $m_E>0,\quad \ $ $ \ n \leq m_I.$ & $m_E>0,\quad \ $ $ \ n>m_I.$ \\
\hline 
1.Projection matrix &  $\OO(1)$ & $\OO(1)$  & $\OO(m_E^{\omega-2}n^2)$ & $\OO(m_E^{\omega-2}n^2)$ \\
\hline
 2.Generating $D$ & $\OO(nz)$ & $\OO(nz)$ & $\OO(n^{\omega-1}z)$ & $\OO(n^{\omega-1}z)$\\
   \hline
 3.Finding $\{L^k\}_{k=1}^z$ &   $\OO(m_In^{\omega -2}z)$   & $\OO(m_I^{\omega -2}nz)$ & $\OO(m_In^{\omega -2}z)$ & $\OO(m_I^{\omega -2}nz)$\\
    \hline
 4.Sampling all $x_{t,j+1}^k$ & $\OO(nz)$  & $\OO(nz)$  &  $\OO(nz)$  &  $\OO(nz)$ \\
 \hline
\end{tabular} 
\end{table}

The following  lemmas will help bound the cost per iteration of MHAR. Lemmas \ref{lemma me=0 n<m} and \ref{lemma me=0 m<n} establish the full dimensional case for ($n \leq m_I$) and ($n > m_I$), respectively. Lemmas \ref{lemma 0<me n<m} and \ref{lemma 0<me m<n} do likewise in the non-full dimensional case for ($n \leq m_I$) and ($n > m_I$), respectively. 

Figure \ref{tree1} summarizes these results as follows.
\begin{figure}[h!]
\centering
\resizebox{.7\textwidth}{!}{
\begin{tikzpicture}[grow=right, sloped]
\node[bag] {Value of $m_E$}
    child {
        node[bag] {$m_I$ \textit{vs} $n$}        
            child {
                node[end, label=right:
                    {$\OO(m_I n^{\omega -2}z)$}] {}
                edge from parent
                node[above] {}
                node[below]  {$n \leq m_I$}
            }
            child {
                node[end, label=right:
                    {$\OO(n^{\omega -1}z)$}] {}
                edge from parent
                node[above] {$n > m_I $}
                node[below]  {}
            }
            edge from parent 
            node[above] {}
            node[below]  {$0 < m_E$}
    }
    child {
        node[bag] {$m_I$ \textit{vs}. $n$}        
        child {
                node[end, label=right:
                    {$\OO(m_In^{\omega -2}z)$}] {}
                edge from parent
                node[above] {}
                node[below]  {$n \leq m_I$}
            }
            child {
                node[end, label=right:
                    {$\OO(m_I^{\omega -2}nz)$}] {}
                edge from parent
                node[above] {$n > m_I $}
                node[below]  {}
            }
        edge from parent         
            node[above] {$m_E=0$}
            node[below] {}
    };
\end{tikzpicture}
}
\caption{Asymptotic behavior of the cost per iteration of MHAR.}
\label{tree1}
\end{figure}
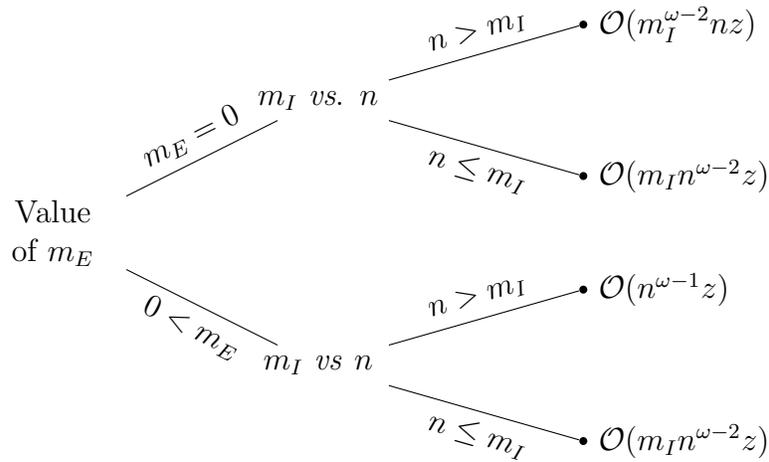

\begin{lemma}\label{lemma me=0 n<m}
Assume $m_E = 0$, $\max\{n,m\} < z$, and $n \leq m_I$. Then, the \textit{cost per iteration} of MHAR is $\OO(m_In^{\omega -2}z)$, which is the number of operations needed for finding all line sets $\{L^k\}_{k=1}^z$.
\end{lemma}
\begin{proof}
First we enumerate the cost of each step of the iteration for $m_E=0$ and $n \leq m_I$ if $\max\{n,m\} < z$:
\begin{enumerate}
     \item By Lemma \ref{lemma projection complexity}, generating $P_{\Delta^E}$ is bounded by $\OO(1)$.
     \item By Lemma \ref{lemma complexity generating D}, generating $D$ is bounded by $\OO(nz)$.
     \item By Lemma \ref{lemma complexity generating L}, generating $\{L^k\}_{k=1}^z$ for $n \leq m_I$ is bounded by $\OO(m_In^{\omega -2}z)$.
     \item By Lemma \ref{lemma complexity sample pick}, generating all new $x_{t,j+1}^k$ is bounded by $\OO(zn)$.
 \end{enumerate}
 
By hypothesis, $0<n\leq m_I$. Then, $nz \leq m_Iz < m_In^{\omega -2}z$, because $\omega \in (2,3]$. Therefore, $\OO(1) \subseteq \OO(nz) \subseteq \OO(m_In^{\omega -2}z)$, where the first term is the complexity of finding the projection matrix (omitted for $m_E=0$), the second one bounds generating $D$ and sampling new points, and the third one is the asymptotic cost of finding all line sets $\{L^k\}_{k=1}^z$.
\end{proof}

\begin{lemma}\label{lemma me=0 m<n}
Assume $m_E = 0$, $\max\{n,m\} < z$, and $n > m_I$. Then, the \textit{cost per iteration} of MHAR is $\OO(nm_I^{\omega -2}z)$, which is the number of operations needed for finding all line sets $\{L^k\}_{k=1}^z$.
\end{lemma}
\begin{proof}
As in the proof of Lemma \ref{lemma me=0 n<m}, the complexity of the projection matrix, generating $D$, and sampling all new $x_{t,j+1}^k$ points is the same, given by $m_E=0$ and $n > m_I$. Hence, the only change is provided by Lemma \ref{lemma complexity generating L}, in which the cost of finding all line sets $\{L^k\}_{k=1}^z$ for  $n > m_I$ is $\OO(nm_I^{\omega -2}z)$.
By hypothesis, $0<m_I$ and $\max\{n,m\} < z$, thus $nz < nm_I^{\omega -2}z$. Therefore, $\OO(1) \subseteq \OO(nz) \subseteq \OO(nm_I^{\omega -2}z)$, where the third term is the cost of finding all line sets $\{L^k\}_{k=1}^z$.
\end{proof}

\begin{corollary}
\label{coro me=0}
Assume $m_E = 0$ and $\max\{n,m\} < z$. Then, the \textit{cost per iteration} of MHAR is bounded by the cost of finding all line sets $\{L^k\}_{k=1}^z$.
\end{corollary}
\begin{proof}
The proof follows from Lemmas \ref{lemma me=0 n<m} and \ref{lemma me=0 m<n}. 
\end{proof}

\noindent We proceed to finding the cost per iteration for the non-full dimensional case $m_E>0$.
\begin{lemma}\label{lma1}
Assume $m_E < n$ and $(m,n)<z$. Then, the cost of calculating the projection matrix $P_{\Delta^E}$ is bounded by the cost of generating $D$.
\end{lemma}
\begin{proof}
By hypothesis $m_E < n $, implying that $m_E^{\omega-2}n^2 < n^{\omega-2}n^2 = n^{\omega}$. Because $n<z$, $n^{\omega} = n^{\omega-1}n  <n^{\omega-1}z $. Combining both inequalities yields $m_E^{\omega-2}n^2 < n^{\omega}< n^{\omega-1}z$. Therefore, $\OO(m_E^{\omega-2}n^2) \subseteq \OO(n^{\omega-1}z)$, where the first term is the complexity of computing $P_{\Delta^E}$ (by Lemma  \ref{lemma projection complexity}), and the second term is the complexity of projecting $H$ in order to obtain $D$ (by Lemma \ref{lemma complexity generating D}).
\end{proof}

\begin{lemma}\label{lemma 0<me n<m}
Assume $m_E>0$, $\max\{n,m\} < z$, and $n \leq m_I$. Then, the \textit{cost per iteration} of MHAR is $\OO(m_In^{\omega -2}z)$, which is the number of operations needed for finding all line sets $\{L^k\}_{k=1}^z$.
\end{lemma}
\begin{proof}
First, we enumerate the cost of each step of the iteration for $m_E>0$, $n \leq m_I$, and $\max\{n,m\} < z$:
\begin{enumerate}
     \item By Lemma \ref{lemma projection complexity}, generating $P_{\Delta^E}$ is bounded by $\OO(m_E^{\omega-2}n^2)$.
     \item By Lemma \ref{lemma complexity generating D}, generating $D$ is bounded by $\OO(n^{\omega -1}z)$.
     \item By Lemma \ref{lemma complexity generating L}, generating $\{L^k\}_{k=1}^z$ for $n \leq m_I$ is bounded by$\OO(m_In^{\omega -2}z)$.
     \item By  Lemma \ref{lemma complexity sample pick}, generating all new $x_{t,j+1}^k$ is bounded by $\OO(zn)$.
 \end{enumerate}
 
Using Lemma \ref{lma1}, the Big-O term for finding $P_{\Delta^E}$ (step 1) is bounded by the term of generating $D$ (step 2). Because $n<m_I$, $n^{\omega-1}z=n^{\omega-2}nz<n^{\omega-2}m_Iz$. Therefore, $\OO(m_E^{\omega-2}n^2)\subseteq \OO(n^{\omega -1}z) \subseteq \OO(m_In^{\omega -2}z)$, which are the respective costs of steps 1, 2, and 3. Furthermore, $nz \leq n^{\omega -2}m_Iz$, implying that step 4 is also bounded by step 3 in terms of complexity. This implies that all the operations above are bounded by the term $\OO(m_In^{\omega -2}z)$, which is the asymptotic complexity of finding all line sets $\{L^k\}_{k=1}^z$.
\end{proof}

\begin{lemma}\label{lemma 0<me m<n}
Assume $m_E>0$, $\max\{n,m\} < z$, and $n>m_I$. Then, the \textit{cost per iteration} of MHAR is $\OO(nm_I^{\omega -2}z)$, which is the number of operations needed for generating $D$.
\end{lemma}
\begin{proof}
As in the proof of Lemma \ref{lemma 0<me n<m}, the cost of the projection matrix, generating $D$, and sampling all new $x_{t,j+1}^k$ points is the same, given by $m_E>0$ and $n > m_I$. Hence, the only change is provided by Lemma \ref{lemma complexity generating L}, in which the cost of finding all line sets $\{L^k\}_{k=1}^z$ for  $n > m_I$ is $\OO(nm_I^{\omega -2}z)$.

By Lemma \ref{lma1}, the Big-O term for finding $P_{\Delta^E}$ is bounded by the term of generating $D$. Because $n>m_I$, $m_I^{\omega-2}nz < n^{\omega-2}nz=n^{\omega-1}z$. Therefore, $\OO(m_E^{\omega-2}n^2) \subseteq \OO(nm_I^{\omega -2}z) \subseteq \OO(n^{\omega -1}z)$, which are the respective costs of the projection matrix, finding all line sets, and generating $D$. Furthermore, $nz \leq n^{\omega -2}nz=n^{\omega -1}z$, implying that the cost of sampling all new $x_{t,j+1}^k$ is also bounded by the cost of generating $D$. This implies that all the operations above are bounded by $\OO(nm_I^{\omega -2}z)$.
\end{proof}

We can now proceed to the main results of the paper, given in Theorem \ref{thmmhar1}.

\begin{theorem}\label{thmmhar1}
If $\max\{n,m\} < z$, then after proper pre-processing and a \textit{warm start}, the \textit{cost per sample} of MHAR is
\begin{equation}
    \begin{cases}
    \ \ \SO(m_In^{\omega+1}), \ \ if \ m_E = 0 \ and  \ n \leq m_I\\
    \ \ \SO(n^{\omega+2}),  \ \ \ \ \ \ if \ m_E = 0 \ and  \ n > m_I\\
    \ \ \SO(m_In^{\omega+1}), \ \ if \ m_E > 0 \ and  \ n \leq m_I\\
    \ \ \SO(m_I^{\omega-2}n^4), \  \ if \ m_E > 0 \ and  \ n > m_I.
    \end{cases}
\end{equation}
\end{theorem}
\begin{proof}
Lemmas \ref{lemma me=0 n<m}, \ref{lemma me=0 m<n}, \ref{lemma 0<me n<m}, and \ref{lemma 0<me m<n} gave the cost per iteration of MHAR for all four cases:

\begin{equation}\label{Thproof}
    \begin{cases}
    \ \ \OO(m_In^{\omega-2}z), \ \ if \ m_E = 0 \ and  \ n \leq m_I\\
    \ \ \OO(n^{\omega-1}z),  \ \ \ \ \ \ if \ m_E = 0 \ and  \ n > m_I\\
    \ \ \OO(m_In^{\omega-2}z), \ \ if \ m_E > 0 \ and  \ n \leq m_I\\
    \ \ \OO(m_I^{\omega-2}nz), \  \ if \ m_E > 0 \ and  \ n > m_I.
    \end{cases} 
\end{equation}

It was stated that each walk from the "padding" is independent about the other ones after a warm-start. Then,
each individual walk has a mixing time of $\SO(n^3)$. Then it suffices to apply the rule for Big-O products between the cost per iteration and the mixing time, and divide the coefficient by the padding parameter $z$, which is the number of points obtained at each iteration. Hence, multiplying each case in Equation (\ref{Thproof}) by $\frac{n^3}{z}$ obtains the desired result. 
\end{proof}

Figure \ref{tree2} graphically depicts the results of the theorem.
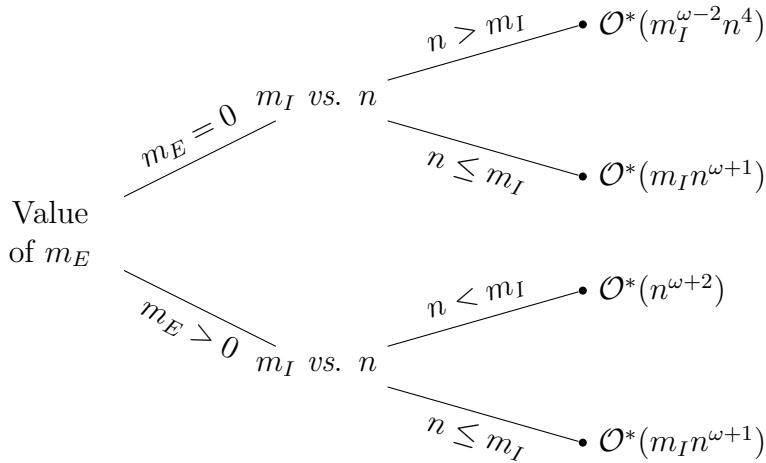
\begin{figure}[h!]
\centering
\resizebox{.7\textwidth}{!}{
\begin{tikzpicture}[grow=right, sloped]
\node[bag] {Value of $m_E$}
    child {
        node[bag] {$m_I$ \textit{vs}. $n$}        
            child {
                node[end, label=right:
                    {$\SO(m_In^{\omega + 1})$}] {}
                edge from parent
                node[above] {}
                node[below]  {$n \leq m_I$}
            }
            child {
                node[end, label=right:
                    {$\SO(n^{\omega +2})$}] {}
                edge from parent
                node[above] {$n < m_I $}
                node[below]  {}
            }
            edge from parent 
            node[above] {}
            node[below]  {$m_E>0$}
    }
    child {
        node[bag] {$m_I$ \textit{vs}. $n$}        
        child {
                node[end, label=right:
                    {$\SO(m_In^{\omega +1})$}] {}
                edge from parent
                node[above] {}
                node[below]  {$n \leq m_I$}
            }
            child {
                node[end, label=right:
                    {$\SO(m_I^{\omega -2}n^4)$}] {}
                edge from parent
                node[above] {$n > m_I $}
                node[below]  {}
            }
        edge from parent         
            node[above] {$m_E=0$}
            node[below] {}
    };
\end{tikzpicture}
}
\caption{Asymptotic behavior of the cost per sample of MHAR after a warm start.}
\label{tree2}
\end{figure}
     
Theorem \ref{thmmhar1} characterizes the cost per sample of MHAR for all parameter values. The theorem shows that MHAR is always at least as efficient as HAR, and more efficient for $\omega \in (2,3)$. Intuitively this is caused by ``padding,'' which permits matrix-to-matrix multiplications instead of isolated matrix-to-vector operations when finding the line sets $L$ or the directions $D$. Furthermore, this approach allows efficient cache usage and state-of-the-art GPU matrix multiplication algorithms.

%%%%%%%%%%%%%%%%%%%%%%%%%%%%%%%%%
%%%%%%%%%%%%%%%%%%%%%%%%%%%%%%%%%
%%%%%%%%%%%%%%%%%%%%%%%%%%%%%%%%%
\section{MHAR Complexity Benchmarks}\label{S:5}
This section benchmarks the asymptotic behavior of MHAR against that for seven state-of-the-art algorithms. Some of these algorithms cover additional convex figures, like spheres or cones. However, we restrict our focus on polytopes because they are the target of MHAR. For in-depth analysis of each algorithm, see \cite{walks}. We prioritize the full-dimensional case ($m=m_I, m_E=0$) because few algorithms are designed for the non-full dimensional scenario and their analysis is outside our scope. Table \ref{table_complex} is adapted from \cite{walks} and includes the notation established in \cite{har_tommi} and \cite{montiel_jd}. The authors of RHCM \cite{RHMC}, John's walk \cite{john_walk}, Vaidya walk, and John walk omitted $m < n$, which is also outside of our scope. Note that John's walk and John walk are different algorithms.

In \S \ref{S:4} we showed that the MHAR has lower cost per sample than the HAR for efficient matrix multiplication algorithms. Furthermore, because the Ball walk \cite{ballwalk} has the same cost per sample as HAR, we can derive the next corollary.
%%%%%%%%%%%%%%%%%%%%%%%%%%%%%%%%
%%%%%%%%%%%%%%%%%%%%%%%%%%%%%%%%
\begin{corollary} \label{cor eff ball}
The \textit{cost per sample} of MHAR is as low as the \textit{cost per sample} of the Ball walk, after a \textit{warm start}, if $\max\{n,m\} < z$. And strictly lower if efficient matrix-to-matrix algorithms are used $\big(\omega \in (2,3)\big)$.
\end{corollary}
%%%%%%%%%%%%%%%%%%%%%%%%%%%%%%%%
\begin{proof}
This follows from comparing Theorem \ref{thmmhar1} against the complexity of the Ball walk.
\end{proof}
%%%%%%%%%%%%%%%%%%%%%%%%%%%%%%%%
The following lemma shows that MHAR has a lower cost per sample than does John's walk.
%%%%%%%%%%%%%%%%%%%%%%%%%%%%%%%%
%%%%%%%%%%%%%%%%%%%%%%%%%%%%%%%%
\begin{lemma} \label{lemma johns walk}
For $\max\{n,m\} < z$, and $n<m$, MHAR has a lower cost per sample than does John's walk after proper pre-processing, \textit{warm start}, and ignoring the logarithmic and error terms.
\end{lemma}
%%%%%%%%%%%%%%%%%%%%%%%%%%%%%%%%
\begin{proof}
Given proper pre-processing, $n \ll m$, and $\max\{n,m\} < z$, then MHAR's cost per sample is $\SO(mn^{\omega + 1})$, and that for John's walk is $\OO(mn^{11} + n^{15})$. Note that $mn^{\omega + 1}\in \OO(mn^{11} + n^{15})$. Therefore, when ignoring the logarithmic and error terms, MHAR has a lower cost per sample.
\end{proof}
%%%%%%%%%%%%%%%%%%%%%%%%%%%%%%%%

\begin{table*}[h!]
\setstretch{1.5} 
\caption{Asymptotic behavior of random walks}
\label{table_complex}
\vspace{.2cm}
\centering
\begin{tabular}{ |p{4.9cm}||p{2.7cm}||p{2.7cm}||p{2cm}|}
 \hline
 \multicolumn{4}{|c|}{Random walks behaviour} \\
 \hline
 \textbf{Walk} &\textbf{Mixing       time} & \textbf{Cost per $\quad \ $ iteration} & \textbf{Cost per $\ $ sample}\\
  \hline
 MHAR with $n > m$  & $n^{3}$    &$m^{\omega -2}nz$ & $m^{\omega -2}n^4$\\
 \hline
  MHAR with $ n \leq m$  & $n^{3}$    &$mn^{\omega -2}z$ & $mn^{\omega +1}$\\
 \hline
 Ball walk    & $n^3$    &$mn$ & $mn^4$\\
  \hline
 HAR  & $n^3$    &$mn$ & $mn^4$\\
  \hline
 Dikin walk with $ n \leq m$  & $mn$    &$mn^{\omega -1}$ & $m^2n^{\omega}$\\
   \hline
 RHCM with $ n \leq m$  & $mn^\frac{2}{3}$    &$mn^{\omega -1}$ & $m^2n^{\omega - \frac{1}{3}}$\\
    \hline
 John's walk with $ n \leq m$  & $n^{7}$    &$mn^4 + n^8$ & $mn^{11} + n^{15}$\\
 \hline
 Vaidya walk with $ n \leq m$   & $m^\frac{1}{2}n^{\frac{3}{2}}$    & $mn^{\omega -1}$ & $m^{1.5}n^{\omega + \frac{1}{2}}$\\
 \hline
 John walk with $ n \leq m$   & $n^\frac{5}{2}log^4(\frac{2m}{n})$    &$mn^{\omega -1}log^2(m)$ & $mn^{\omega + \frac{3}{2}}$\\
 \hline
\end{tabular}
    \begin{tablenotes}
      \footnotesize 
      \item The table contains the upper bounds on the cost per sample (after a warm start) for various random walk algorithms applied to polytopes. In the case of MHAR, $\max\{n,m\} < z$ is assumed. For simplicity, we ignore the logarithmic terms in the cost per sample. We also avoid giving bounds in terms of the condition number of the set for MHAR, Ball walk, and HAR, because this condition number is bounded by $n$ after proper pre-processing. 
    \end{tablenotes}
\end{table*}

In the regime of $n \ll m$, the overall upper bound complexity for the cost per sample is represented by John walk $\ll$ Vaidya walk $\ll$ Dikin walk \cite{walks}. We now show that for $n \ll m$, MHAR has a lower cost per sample than does John walk.
%%%%%%%%%%%%%%%%%%%%%%%%%%%%%%%%

%%%%%%%%%%%%%%%%%%%%%%%%%%%%%%%%
\begin{lemma}\label{lemma mhar john}
For $\max\{n,m\} < z$ and the regime $n \ll m$, MHAR has a lower cost per sample than does the John walk after proper pre-processing, warm start, and ignoring logarithmic and error terms.
\end{lemma}
%%%%%%%%%%%%%%%%%%%%%%%%%%%%%%%%
\begin{proof}
From proper pre-processing, $n \ll m$, and $\max\{n,m\} < z$ , MHAR's cost per sample is $\SO(mn^{\omega + 1})$ and that for John walk is $\OO(mn^{\omega + \frac{3}{2}})$. Note that $mn^{\omega + 1} \in \OO(mn^{\omega + \frac{3}{2}})$. Therefore when ignoring the logarithmic and error terms, MHAR has a lower cost per sample.
\end{proof}
%%%%%%%%%%%%%%%%%%%%%%%%%%%%%%%%
%%%%%%%%%%%%%%%%%%%%%%%%%%%%%%%%
%%%%%%%%%%%%%%%%%%%%%%%%%%%%%%%%
\begin{corollary}\label{cor mhar others}
For $\max\{n,m\} < z$ and the regime $n \ll m$, then MHAR $\ll$ John Walk $\ll$ Vaidya walk $\ll$ Dikin walk after proper pre-processing, warm start, and ignoring logarithmic and error terms.
\end{corollary}
%%%%%%%%%%%%%%%%%%%%%%%%%%%%%%%%
\begin{proof}
This follows from Lemma \ref{lemma mhar john}.
\end{proof}
%%%%%%%%%%%%%%%%%%%%%%%%%%%%%%%%
 
 We proceed to compare MHAR and RHMC for the regime $n^{1+\frac{1}{3}} \ll m$.
\begin{lemma}\label{lemma mhar rhmc}
For  $\max\{n,m\} < z$ and $n^{1+\frac{1}{3}} \ll m$, then MHAR $\ll$ RHMC after proper pre-processing, warm start and ignoring logarithmic and error terms.\end{lemma}
\begin{proof}
From proper pre-processing, $n \ll m$, and $n,m < z$, MHAR's cost per sample is $\SO(mn^{\omega + 1})$, and RHMC's is $\OO(m^2n^{\omega - \frac{1}{3}})$. Note that $mn^{\omega + 1} \in \OO(m^2n^{\omega - \frac{1}{3}})$, because $n^{1+\frac{1}{3}} \ll m$. Therefore, when ignoring the logarithmic and error terms, MHAR has a lower cost per sample.
\end{proof}

From corollaries \ref{cor eff ball} and \ref{lemma johns walk}, MHAR $\ll$ Ball walk and MHAR $\ll$ HAR, regardless of the regime between $m$ and $n$. And MHAR $\ll$ John's Walk for the regime $n\leq m$. From corollary \ref{cor mhar others}, MHAR $\ll$ John Walk $\ll$ Vaidya walk $\ll$ Dikin walk if $n < m$. Finally, by Lemma \ref{lemma mhar rhmc}, if $n^{1+\frac{1}{3}} \ll m$, then MHAR $\ll$ RHMC. 

Then, if $n^{1+\frac{1}{3}} \ll m$ we have an analytic guarantee that MHAR has a lower cost per sample than all of the other algorithms in Table \ref{table_complex}. Moreover, empirical tests show that MHAR is faster than all of the other algorithms in Table \ref{table_complex} for regimes other than  $n^{1+\frac{1}{3}} \ll m$.
 
%%%%%%%%%%%%%%%%%%%%%%%%%%%%%%%%%%%%%%%%%%%%%%%%%%%%%%%%%%%%%%%%%%%%
%%%%%%%%%%%%%%%%%%%%%%%%%%%%%%%%%%%%%%%%%%%%%%%%%%%%%%%%%%%%%%%%%%%%
%%%%%%%%%%%%%%%%%%%%%%%%%%%%%%%%%%%%%%%%%%%%%%%%%%%%%%%

%%%%%%%%%%%%%%%%%%%%%%%%%%%%%%%%%%%%%%%%%%%%%%%%%%%%%%%%%%%%%%%%%%%%%%%%%%%%%%%%%%%%%%%%%%%%%%%%%%%%%%%%%%%%%%%%%%%%%%%%%%%%%%%%%%%%%%%
%%%%%%%%%%%%%%%%%%%%%%%%%%%%%%%%%%%%%%%%%%%%%%%%%%%%%%%%%%%%%%%%%%%%%%%%%%%%%%%%%%%%%%%%%%%%%%%%%%%%%%%%%%%%%%%%%%%%%%%%%%%%%%%%%%%%%%%
\section{MHAR Empirical Test}\label{S:6}
This section details a series of experiments to compare MHAR against the \textit{hitandrun} library used by \cite{har_tommi}. We compare the running times in simplexes and hypercubes of different dimensions and for various values of the padding hyper-parameter $z$. We also test the robustness of MHAR by conducting empirical analyses similar to those in \cite{har_tommi}. 
MHAR experiments were run in a Colab Notebook equipped with an Nvidia P100 GPU, and a processor Intel\textsuperscript{\textregistered} Xeon\textsuperscript{\textregistered} CPU running at 2.00 GHz, and 14 GB of RAM.
Due to its apparent incompatibility with the Colab Notebook, the \textit{hitandrun} experiments were run in a $<$device$>$ equipped with an Intel\textsuperscript{\textregistered} Core\texttrademark\ i7-7700HQ CPU running at 2.80 GHz and 32 GBs of RAM. All experiments used 64 bits of precision.

We formally define the $n\mbox{-simplex}$ and the $n\mbox{-hypercube}$ as
\begin{align}
n\mbox{-simplex} &= \{ x \in \R^n \| \sum x_i = 1, x \geq 0 \},\\
n\mbox{-hypercube} &= \{ x \in \R^n \| x \in [-1,1]^n\}.
\end{align}

%%%%%%%%%%%%%%%%%%%%%%%%%%%%%%%%%%%%%%%%%%%%%%%%%%%%%%%%%%%%%%%%%%%%
\subsection{The Code}
The MHAR code was developed using python, and the Pytorch library was chosen because of its flexibility, power, and popularity \cite{paszke2019pytorch}. Pytorch also works in a CPU without need of a GPU, although the latter is more suitable for large samples in high dimensions.  
The MHAR experiments were performed without observing any numerical instabilities, and the maximum error found for the inversion matrix was on the order $1e$-$16$, which is robust enough for most applications. Operations such as matrix inversion, random number generation, matrix-to-matrix multiplication, and point-wise operations were carried out in the GPU. The only operations that needed to be carried out in the CPU were reading the constraints and saving the samples to disk. 

For the rest of this section, the acronyms MHAR and HAR refer to the actual implementations and not the abstract algorithms. The code is available in \url{https://github.com/uumami/mhar_pytorch}.

%%%%%%%%%%%%%%%%%%%%%%%%%%%%%%%%%%%%%%%%%%%%%%%%%%%%%%%%%%%%%%%%%%%%
\subsection{The padding}
The padding hyper-parameter $z$ determines the number of simultaneous walks the algorithm performs. We generated 10 MHAR runs for each dimension (5, 25, 50, 100, 500, 1000) and each padding value ($z$) on simplexes and hypercubes. At each run we calculated the average samples per second as follows:
\[ 
Avg. \ Samples  \ per \  Second = \frac{Total \ Samples}{Time} = \frac{z \times \varphi \times T}{Time}.
\]
For example, $z$ might equal 100, the thinning parameter $\varphi$ might equal 30,000, and the number of iterations $T$ might equal 1, which would yield $3,000,000$ samples. If the experiment took 1,000 seconds, the average samples per second would be $3,000$.

Figures \ref{padding_times_S} and \ref{padding_times_H} show box-plots for the experiments in dimensions 5 and 1000 for the simplex and the hypercube, respectively. The box-plots for the the simplex and the hypercube in dimensions 25, 50, 100 and 500 can be found in Figures \ref{padding_times_SA} and \ref{padding_times_HA} in \ref{BP1}.
\begin{figure}[h!]
\centering
\subfigure[Unit simplex in dimension 5.]{
\label{subfig:pt_S_5}
\includegraphics[scale=.33,viewport=20 0 650 500,clip]{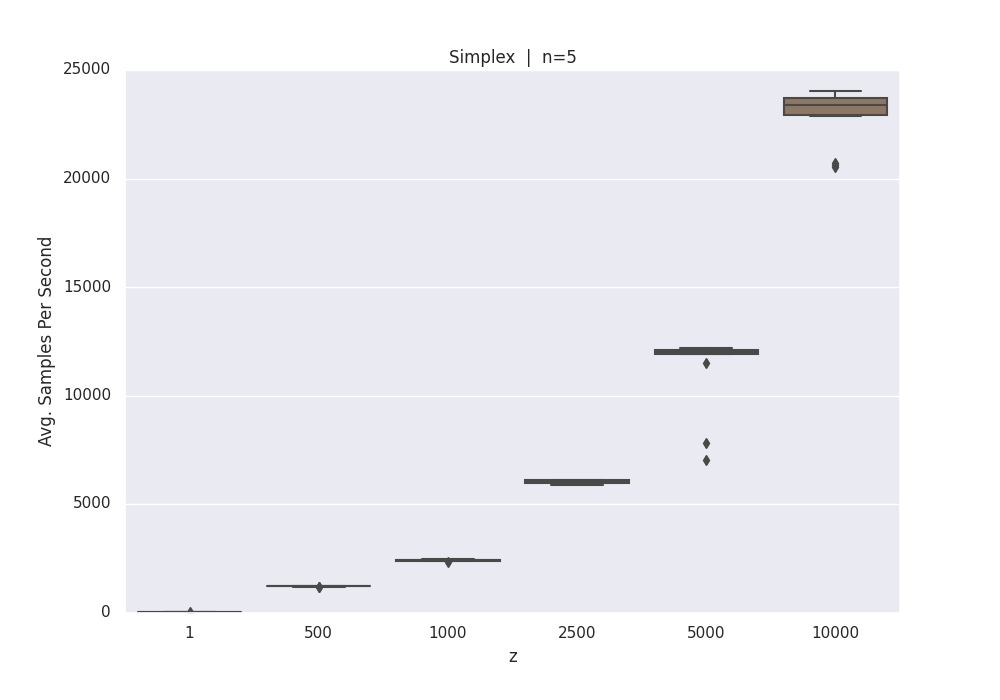}
}
\subfigure[Unit Simplex in dimension 1000.]{
\label{subfig:pt_S_1000}
\includegraphics[scale=.33,viewport=20 0 650 500,clip]{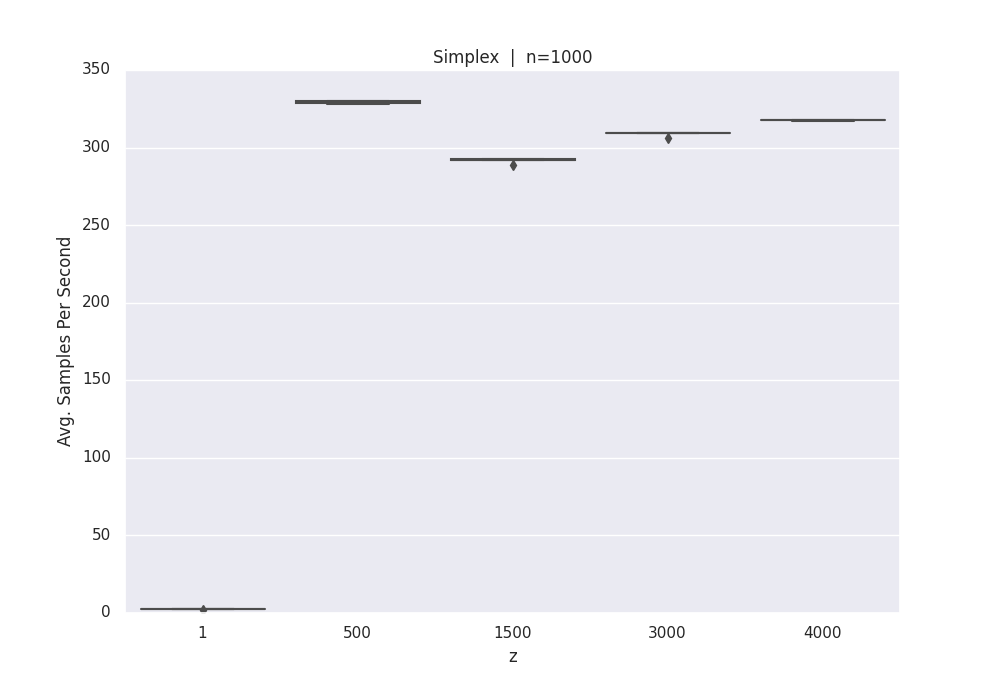}
}
\caption{Box-plots for simplexes comparing padding behavior . In the y-axis the average samples per second are in thousands for different values of the padding parameter $z$.}
\label{padding_times_S}
\end{figure} 

The box in the box-plots show the 25\%, 50\%, and 75\% percentiles. The diamonds mark outliers, and the upper and lower limits mark the maximum and minimum values without considering outliers. For small values of $z$, larger padding yielded more average samples per second. However, for some dimensions in the simplex and the hypercube, there was a value of $z$ for which efficiency was lower. We conjecture that at some point large values of $z$ could cause memory contention in the GPU.
\begin{figure}[h!]
\centering
\subfigure[Hypercube in dimension 5.]{
\label{subfig:pt_H_5}
\includegraphics[scale=.33,viewport=20 0 650 500,clip]{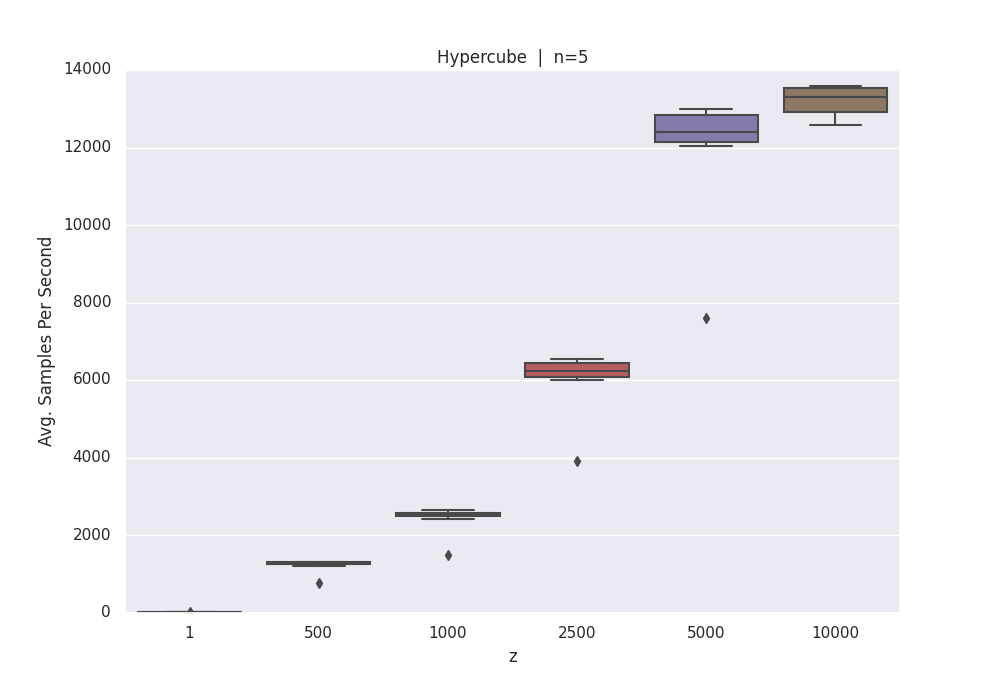}
}
\subfigure[Hypercube in dimension 1000.]{
\label{subfig:pt_H_1000}
\includegraphics[scale=.33,viewport=20 0 650 500,clip]{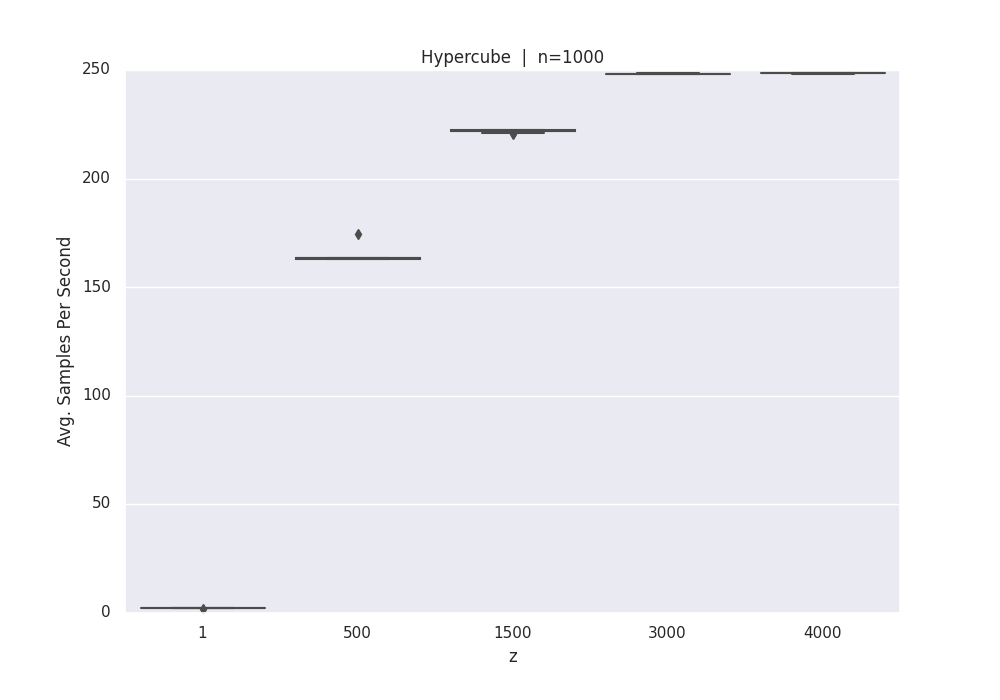}
}
\caption{Box-plots for hypercubes comparing padding behavior. In the y-axis the average samples per second are in thousands for different values of the padding parameter $z$.}
\label{padding_times_H}
\end{figure} 

\subsection{Performance Test MHAR vs HAR}

To compare MHAR and HAR we generated 10 simulations for different dimensions, and two types of polytopes (simplex and hypercubes). For the simplex we tested dimensions: 5, 25, 50, 100, and 250, and for the hypercube we tested dimensions: 5, 25, 50, 100, 500, and 1000. The \textit{hitandrun} routines for sampling the simplex exhibited an extreme drop in performance at dimensions higher than $100$ and memory contention at dimensions higher than $300$. 

For \textit{hitandrun}, the total number of samples equals number of iterations times the thinning parameter. Because \textit{hitandrun} does not make use of the GPU, the times are dependent on the CPU. Before running a given combination of convex body and dimension in MHAR, we selected the padding hyper-parameter $z^*$ that had the highest average sampled points per second according to our padding experiments. So the $z^*$ can differ by dimension. We used $\varphi=30,000$ and $T=1$. Table \ref{mhar_speedup} summarizes the results.
 \begin{table*}[h!]
\caption{Performance of MHAR versus HAR for the optimal value of $z^*$}
\label{mhar_speedup}
\vspace{.2cm}
\centering 
\resizebox{.9\hsize}{!}{
\begin{tabular}{lrrrrrrr}
\hline
& & & & & \multicolumn{2}{c}{Avg. Samples Per Second} \\
\cline{5-6}
\cline{7-8}
Figure &     n &      $z$ &  Performance ratio $\quad \quad $    &     MHAR mean &   HAR mean &     MHAR Std. Dev. &    HAR Std. Dev. \\
       &       &          & (MHAR mean / HAR mean)                 &               &            &                    &    \\ 
    \hline
 Hypercube &     5 &  10,000 &        14.18 &             13,206,089.93 & 931,368.92 &   376,068.96 &  57,727.69 \\
 Hypercube &    25 &   5,000 &        29.05 &             10,839,474.35 & 373,127.77 & 1,236,619.81 &  77,786.96 \\
 Hypercube &    50 &   2,500 &        21.85 &              5,151,516.81 & 235,742.22 &   612,241.73 &  20,636.30 \\
 Hypercube &   100 &   4,000 &       116.77 &              4,363,525.70 &  37,367.93 &    10,619.65 &   1,486.54 \\
 Hypercube &   500 &   4,000 &        95.21 &                621,554.24 &   6,528.56 &       782.70 &     157.76 \\
 Hypercube &  1,000 &   4,000 &       248.32 &                248,513.69 &   1,000.79 &       182.97 &      18.15 \\
   Simplex &     5 &  10,000 &        23.14 &             22,878,783.33 & 988,580.92 & 1,258,481.83 & 126,254.73 \\
   Simplex &    25 &  10,000 &     1,343.58 &             24,338,761.06 &  18,114.90 &   168,300.75 &     409.27 \\
   Simplex &    50 &  10,000 &    12,630.89 &             13,425,900.57 &   1,062.94 &    16,403.51 &      17.33 \\
   Simplex &   100 &   3,000 &   128,348.67 &              7,255,837.08 &      56.53 &   135,616.62 &       0.88 \\
   Simplex &   250 &   4,000 & 2,551,224.17 &              2,656,449.22 &       1.04 &     4,440.59 &       0.00 \\
\hline
\end{tabular}
}
\end{table*}
 %%%%%%%%%%%%%%%%%%%%%%%%%%%%%%%%%%%%%%%%%%%%%%%%%%%%%%%%%%%%%%%%
 
Table \ref{mhar_speedup} shows substantial performance gains for MHAR. For the simplex, the gains were greater at higher dimensions. The performance ratio (average samples per second for MHAR divided by that for HAR) was $23$ for $n=5$ and $2.5$ million for $n=250$. For the hypercube, performance gain for MHAR was also greater at higher dimensions. Nevertheless, the performance ratio was $14$ for $n=5$ and $248$ for $n=1,000$.
 
 In order to test the limits of our implementation, we conducted an additional set of experiments for lower and higher dimensions and different padding parameters. We present these results in \ref{App1}.
 
\subsection{Independence Test}

To asses the convergence of MHAR to a uniform distribution, we conducted Friedman-Rafsky two-sample Minimum Spanning Tree (MST) test \cite{mst}, as was done in \cite{har_tommi}. The test compares an obtained sample $S$ (MHAR) with a sample $U$ from the target distribution. The test defines an MST for $S$ and $U$ by counting the number of within- and across-sample edges to assess if both samples come from the same distribution. The statistic from the tests yields a \textit{z-value} for the null hypothesis: ``Both samples are drawn from the same distribution.'' Authors in \cite{har_tommi} establish a threshold of $-1.64 \leq$ \textit{z-value} to accept the null hypothesis.

A uniform sample $U$ can quickly be drawn from the hypercube or the simplex \cite{simplex_sample} using known statistical methods. 
We generated 10 simulations in simplexes and hypercubes in dimensions: $5$, $15$, $25$, and $50$, for a total of 80 simulations. We used a single padding parameter ($z$) of $1000$; and a "burning rate" ($\varphi$) of $(n-1)^3$ for the simplex, and $n^3$ for the hypercube. Each simulation draw a total of $5000$ samples that were compared to an independently generated sample $U$ each time. 

Figure \ref{z_box} shows the results from the experiments. The red dashed line represents the threshold of $-1.64 \leq$ \textit{z-value}. All simulations where above the expected threshold with the exception of one single experiment for the simplex in dimension 25. This experiments suggests that MHAR mixes fast from any starting point, supporting the uniform sample hypothesis.
\begin{figure}[h]
\centering
\subfigure[Simplex.]{
\label{subfig:pt_Sim}
\includegraphics[scale=.33,viewport=20 0 650 500,clip]{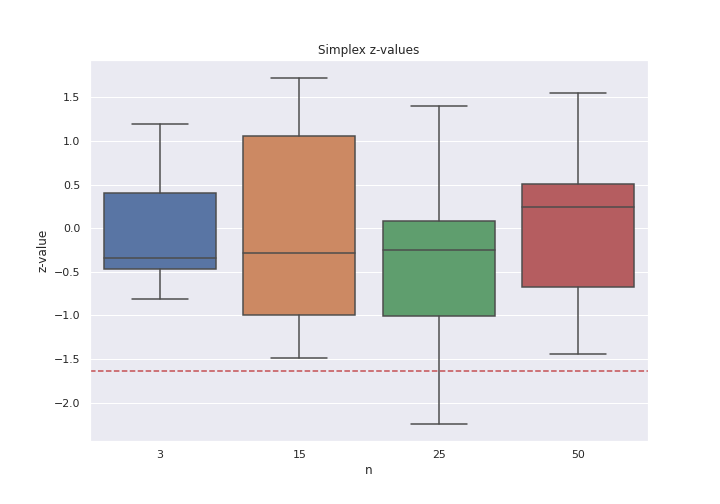}
}
\subfigure[Hypercube.]{
\label{subfig:pt_Hc}
\includegraphics[scale=.33,viewport=20 0 650 500,clip]{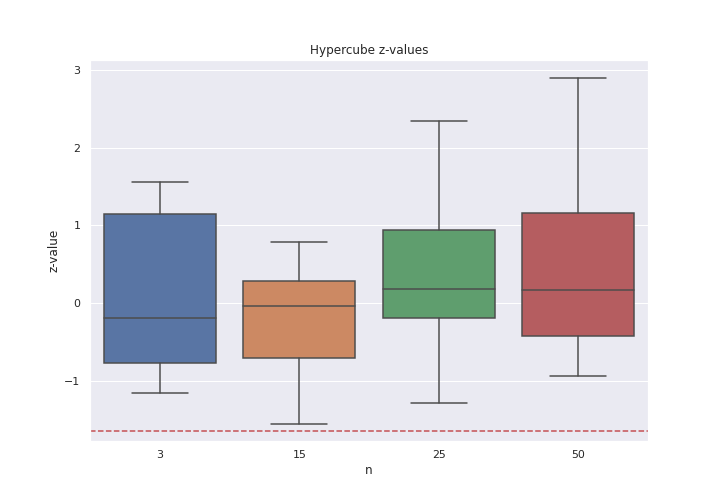}
}
\caption{Friedman-Rafsky two-sample MST tests.}
\label{z_box}
\end{figure}
 
\section{Conclusions}\label{S:7}

MHAR showed sustainable performance improvements over HAR while having a robust uniform sampling. We hope that this technical advances move the scientific community towards simulation approaches to complement the already established analytical solutions. Our contribution was in creating the MHAR, analyzing its asymptotic behavior in terms of complexity and convergence, alongside a robust and easy to use implementation ready for deployment, including the cloud. Our implementation is substantially faster than existing libraries, especially for bigger dimensions. Additionally, we showed the versatility that Deep Learning frameworks, like Pytorch, can bring to support research.

We would like to emphasize the relevance of this work as a cornerstone to exploratory-optimization algorithms. The speedups we present in high dimensions makes it possible for many new practical applications to become a normal trend, expanding the range of solutions that engineering can provide. In particular, our previous work in Decision Analysis, Optimization, Game Theory, and Ambiguity Optimization will be significantly improved with this tool, and we think that many practitioners and researchers will be benefit as well.   

Our implementation could be extended to multiple GPUs, possibly distributed. This will allow us to sample even larger polytopes using cloud architectures. 
Given the speed up results, a bounding approach for more general convex figures alongside accept-and-reject methods is worth exploring, especially for volume calculations.

\section*{Acknowledgments}
This work was supported by the National Council of Science and Technology of Mexico (CONACYT) and the National System of Researchers (SNI) under Luis V. Montiel, Grant No. 259968. In addition, we also acknowledge Dr. Fernando Esponda, Dr. Jose Octavio Gutierrez, and Dr. Rodolfo Conde for their support and insight in the development of this work.

%%%%%%%%%%%%%%%%%%%%%%%%%%%%%%%%%%%%%%%%%%%%%%%%%%%%%%%%%%%%%%%%%%%%
%%%%%%%%%%%%%%%%%%%%%%%%%%%%%%%%%%%%%%%%%%%%%%%%%%%%%%%%%%%%%%%%%%%%
%%%%%%%%%%%%%%%%%%%%%%%%%%%%%%%%%%%%%%%%%%%%%%%%%%%%%%%%%%%%%%%%%%%%
\setstretch{.9} 

\bibliographystyle{elsarticle-num-names}
\bibliography{sample.bib}

 %%%%%%%%%%%%%%%%%%%%%%%%%%%%%%%%%%%%%%%%%%%%%%%%%%%%%%%%%%%%%%%%%%%%
%%%%%%%%%%%%%%%%%%%%%%%%%%%%%%%%%%%%%%%%%%%%%%%%%%%%%%%%%%%%%%%%%%%%
%%%%%%%%%%%%%%%%%%%%%%%%%%%%%%%%%%%%%%%%%%%%%%%%%%%%%%%%%%%%%%%%%%%%
\newpage
\appendix 
\setstretch{1.5} 

\section{Additional Optimal Padding Experiments}\label{BP1}

Here we present the results for different padding parameters using 10 MHAR runs for each dimension (25, 50, 100, 500) on simplexes and hypercubes. Figure \ref{padding_times_SA} shows the box-plots for simplexes while Figure \ref{padding_times_HA} shows the box-plots for hypercubes. 
\begin{figure}[h!]
\centering
\subfigure[Unit simplex in dimension 25.]{\label{subfig:pt_S_25}\includegraphics[scale=.34,viewport=20 0 650 500,clip]{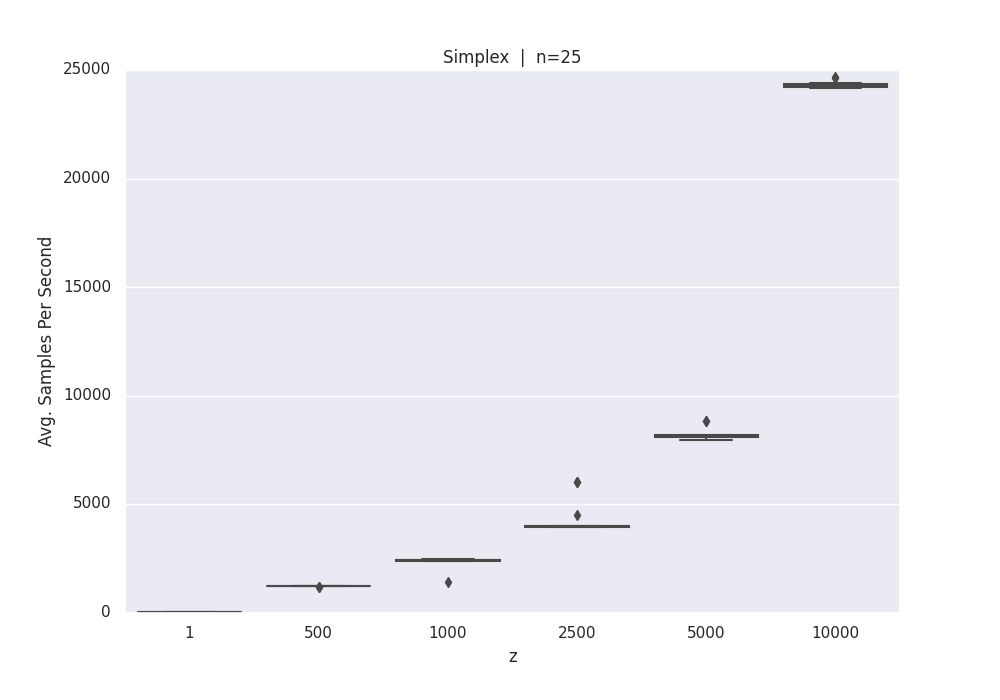}}
\subfigure[Unit simplex in dimension 50.]{\label{subfig:pt_S_50}\includegraphics[scale=.34,viewport=20 0 650 500,clip]{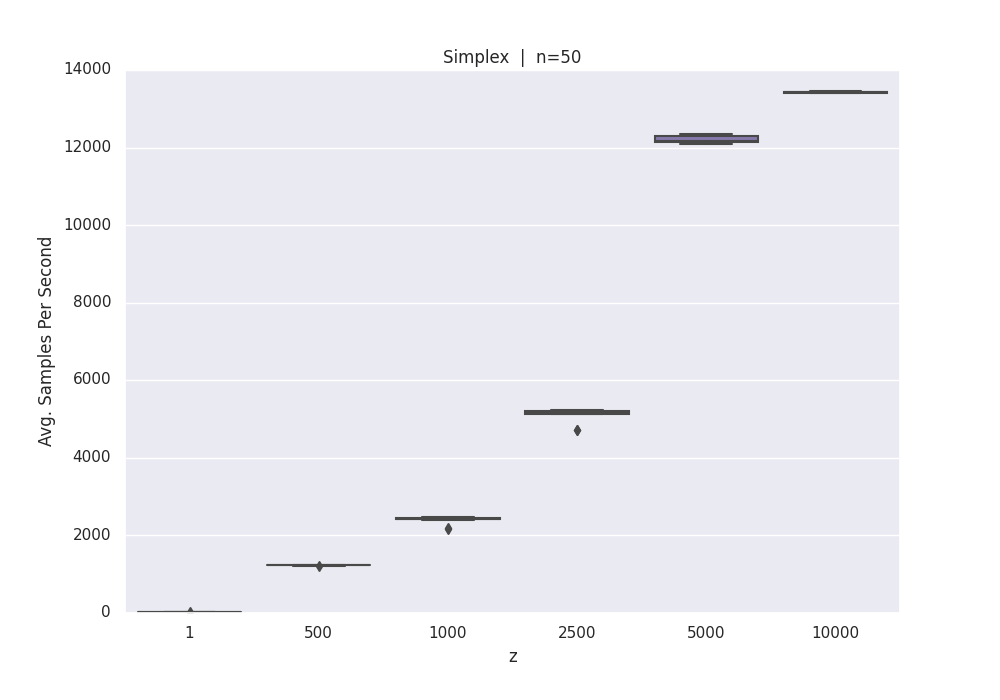}}
\subfigure[Unit simplex in dimension 100.]{\label{subfig:pt_S_100}\includegraphics[scale=.34,viewport=20 0 650 500,clip]{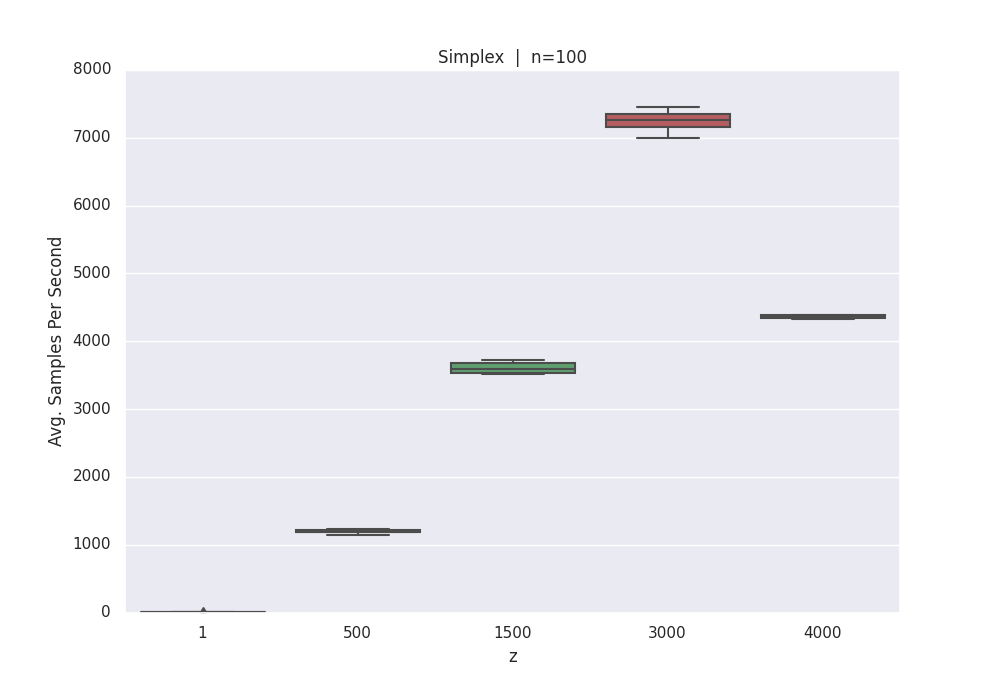}}
\subfigure[Unit simplex in dimension 500.]{\label{subfig:pt_S_500}\includegraphics[scale=.34,viewport=20 0 650 500,clip]{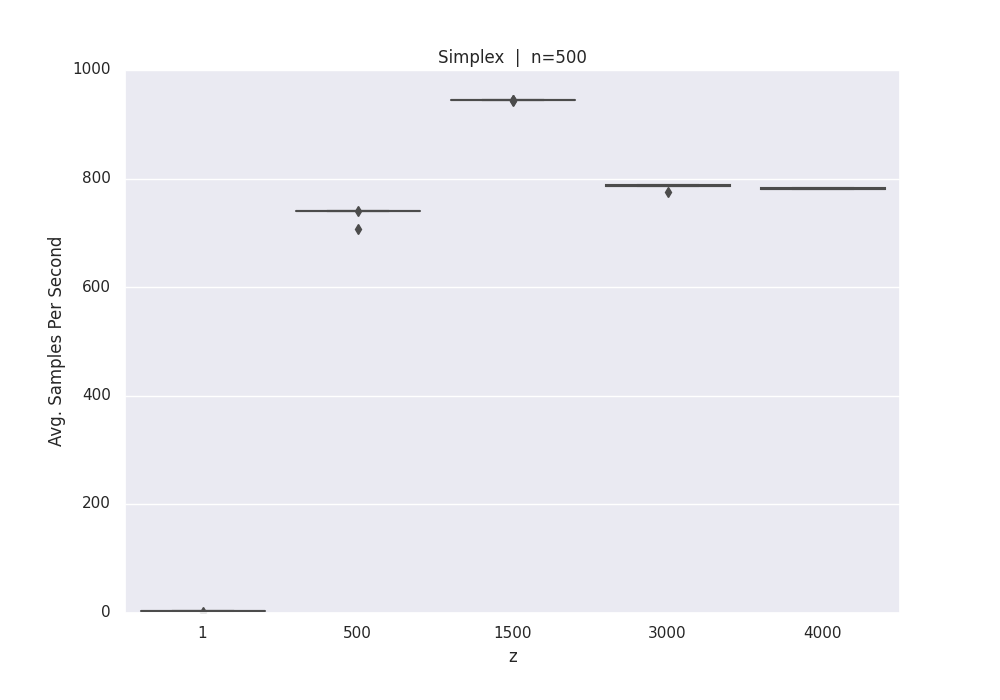}}
\caption{Box-plots for simplexes comparing padding behavior. In the y-axis the average samples per second are in thousands for different values of the padding parameter $z$.}
\label{padding_times_SA}
\end{figure} 

The box in the boxplots show the 25\%, 50\%, and 75\% percentiles. The diamonds mark outliers, and the upper and lower limits mark the maximum and minimum values without considering outliers.

\begin{figure}[h!]
\centering
\subfigure[Hypercube in dimension 25.]{\label{subfig:pt_H_25}\includegraphics[scale=.34,viewport=20 0 650 500,clip]{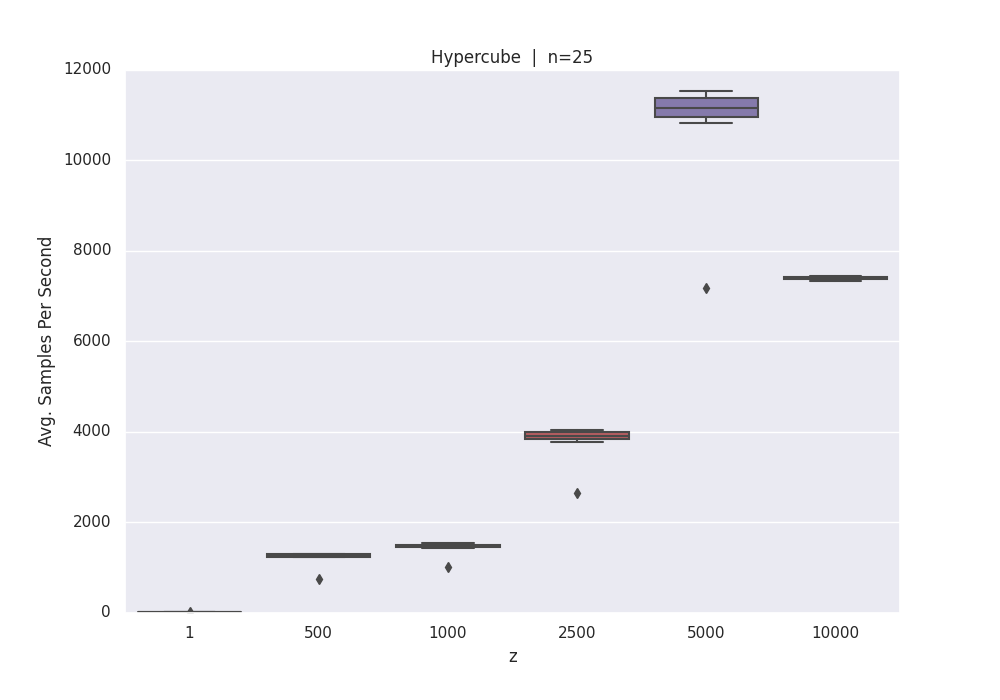}}
\subfigure[Hypercube in dimension 50.]{\label{subfig:pt_H_50}\includegraphics[scale=.34,viewport=20 0 650 500,clip]{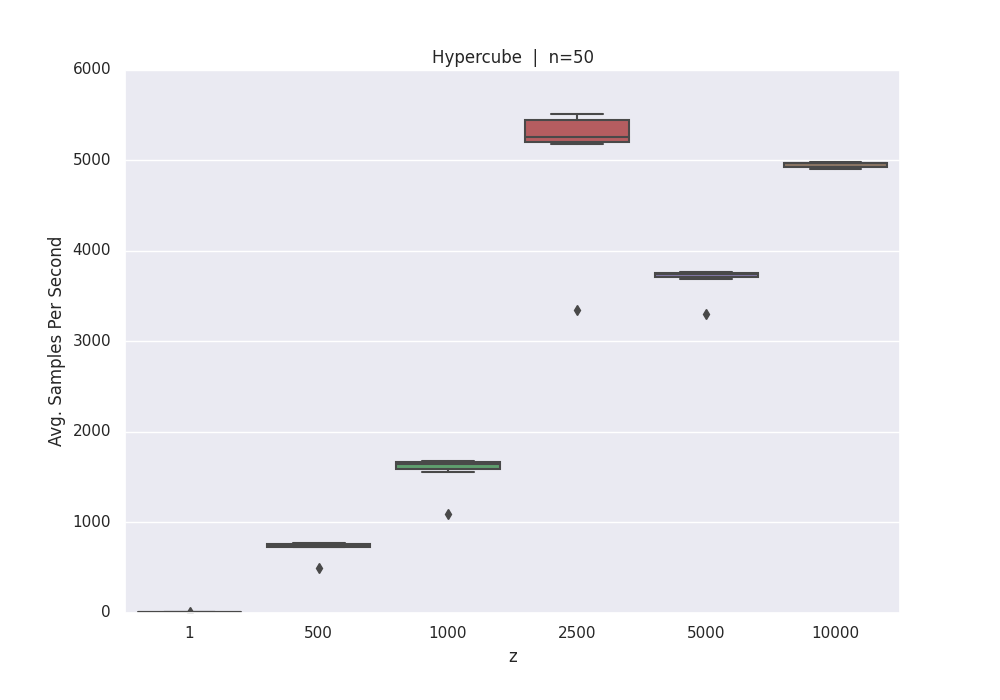}}
\subfigure[Hypercube in dimension 100.]{\label{subfig:pt_H_100}\includegraphics[scale=.34,viewport=20 0 650 500,clip]{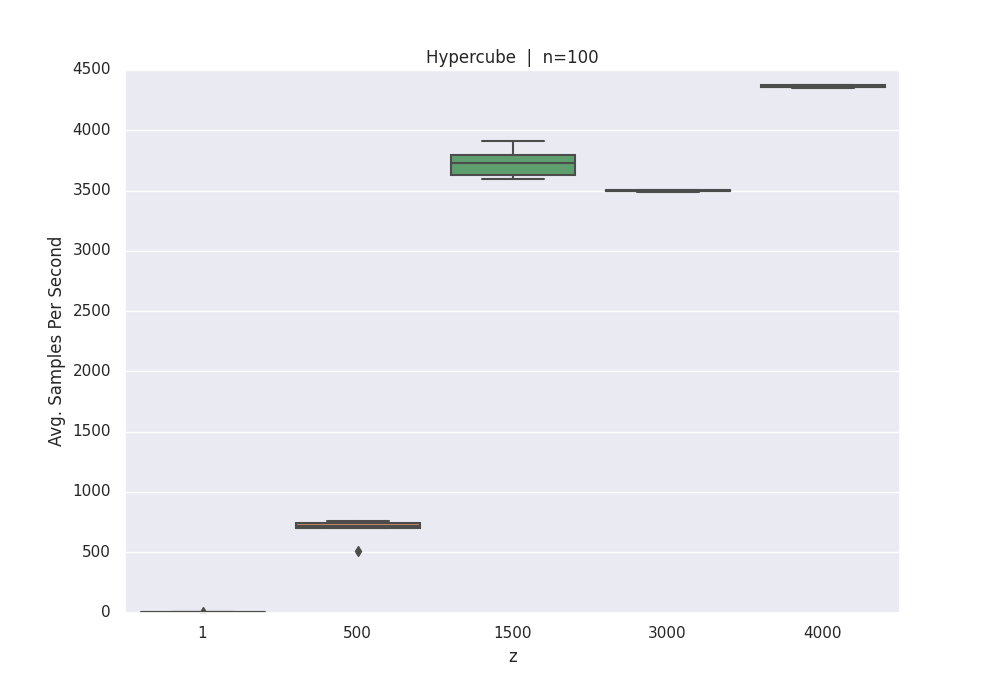}}
\subfigure[Hypercube in dimension 500.]{\label{subfig:pt_H_500}\includegraphics[scale=.34,viewport=20 0 650 500,clip]{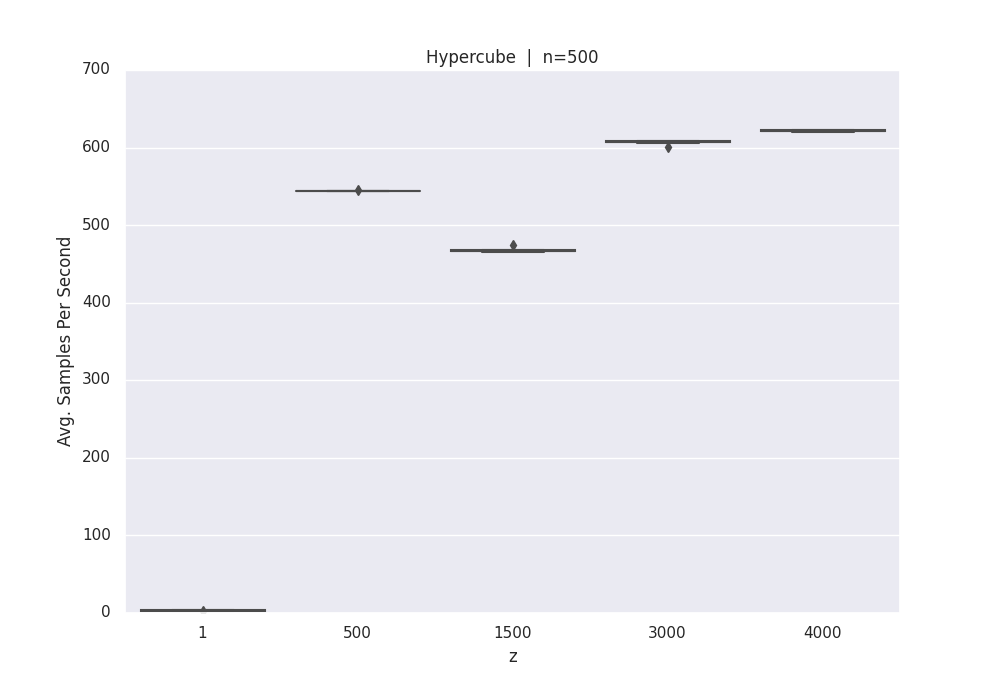}}
\caption{Box-plots for hypercube comparing padding behavior. In the y-axis the average samples per second are in thousands for different values of the padding parameter $z$.}
\label{padding_times_HA}
\end{figure} 

\newpage
\section{Additional Performance Experiments}\label{App1}

Here we present some additional experiments of the fitness of the MHAR. Table \ref{mhar_table_times} reports the running times and the average sampled points per second for the best values of $z$ for each combination of figure and dimension. For each combination, we conducted the experiment 10 times. 
Table \ref{mhar_table_times} shows that average samples per second is lower for higher dimensions, due to the curse of dimensionality. However, the performance of MHAR is outstanding.

\begin{table*}[t]
\caption{Samples Per Second of the MHAR.}
\vspace{.2cm}
\label{mhar_table_times}
\centering 
\resizebox{.9\hsize}{!}{
\begin{tabular}{lrrrrrrr}
\hline
& & & & \multicolumn{2}{c}{Avg. Samples Per Second} & \multicolumn{2}{c}{Running Time (seconds)} \\
\cline{5-6}
\cline{7-8}
    Figure &     n &   $z$ & Total Samples &                      Mean &          Std. Dev. &         Mean &  Std. Dev. \\
    \hline
 Hypercube &     3 &  10,000 &      300,000,000 &             25,357,073.87 &   675,444.40 &        11.84 & 0.32 \\
 Hypercube &     5 &  10,000 &      300,000,000 &             13,206,089.93 &   376,068.96 &        22.73 & 0.66 \\
 Hypercube &    15 &  10,000 &      300,000,000 &             25,344,794.68 &   655,021.48 &        11.84 & 0.31 \\
 Hypercube &    25 &   5,000 &      150,000,000 &             10,839,474.35 & 1,236,619.81 &        14.07 & 2.28 \\
 Hypercube &    50 &   2,500 &       75,000,000 &              5,151,516.81 &   612,241.73 &        14.83 & 2.54 \\
 Hypercube &   100 &   4,000 &      120,000,000 &              4,363,525.70 &    10,619.65 &        27.50 & 0.07 \\
 Hypercube &   250 &   3,000 &       90,000,000 &              1,219,419.53 &     8,630.27 &        73.81 & 0.53 \\
 Hypercube &   500 &   4,000 &      120,000,000 &                621,554.24 &       782.70 &       193.06 & 0.24 \\
 Hypercube &  1,000 &   4,000 &      120,000,000 &                248,513.69 &       182.97 &       482.87 & 0.36 \\
 Hypercube &  2,500 &   1,500 &       15,000,000 &                 50,808.74 &        15.02 &       295.22 & 0.09 \\
 Hypercube &  5,000 &   1,000 &       10,000,000 &                 16,161.69 &         5.92 &       618.75 & 0.23 \\
   Simplex &     3 &  10,000 &      300,000,000 &             19,795,014.21 & 2,628,558.29 &        15.38 & 1.81 \\
   Simplex &     5 &  10,000 &      300,000,000 &             22,878,783.33 & 1,258,481.83 &        13.15 & 0.77 \\
   Simplex &    15 &  10,000 &      300,000,000 &             24,269,548.32 &   302,854.48 &        12.36 & 0.16 \\
   Simplex &    25 &  10,000 &      300,000,000 &             24,338,761.06 &   168,300.75 &        12.33 & 0.08 \\
   Simplex &    50 &  10,000 &      300,000,000 &             13,425,900.57 &    16,403.51 &        22.34 & 0.03 \\
   Simplex &   100 &   3,000 &       90,000,000 &              7,255,837.08 &   135,616.62 &        12.41 & 0.23 \\
   Simplex &   250 &   4,000 &      120,000,000 &              2,656,449.22 &     4,440.59 &        45.17 & 0.08 \\
   Simplex &   500 &   1,500 &       45,000,000 &                944,784.52 &       583.24 &        47.63 & 0.03 \\
   Simplex &  1,000 &    500 &       15,000,000 &                329,315.49 &       556.62 &        45.55 & 0.08 \\
   Simplex &  2,500 &    500 &        5,000,000 &                 77,312.01 &     3,045.62 &        64.78 & 2.86 \\
   Simplex &  5,000 &   1,000 &       10,000,000 &                 22,437.63 &        62.27 &       445.68 & 1.25 \\
\hline
\end{tabular}
}
     \begin{tablenotes}
      \footnotesize 
      \item Note: The table contains the performance statistics obtained during the MHAR experiments for the best possible value of $z$ we could find.
    \end{tablenotes}
\end{table*}

\end{document}